\documentclass[10pt,twocolumn,twoside]{IEEEtran}
\usepackage{cite}
\usepackage{graphicx}
\usepackage{subcaption}
\usepackage[cmex10]{amsmath}
\usepackage{algpseudocode}
\usepackage{algorithmicx} 
\usepackage{array}
\usepackage{url}
\usepackage{caption}
\usepackage{epsfig}
\usepackage{amsfonts,amssymb,multirow,bigstrut,booktabs,ctable,latexsym}

\usepackage{cuted}
\usepackage{flushend}

\usepackage[mathscr]{eucal}
\usepackage{verbatim}
\usepackage{amsthm}
\usepackage{algorithm}

\begin{document}

\newtheorem{definition}{Definition}
\newtheorem{lemma}{Lemma}
\newtheorem{theorem}{Theorem}
\newtheorem{example}{Example}
\newtheorem{proposition}{Proposition}
\newtheorem{remark}{Remark}
\newtheorem{assumption}{Assumption}
\newtheorem{corrolary}{Corrolary}
\newtheorem{property}{Property}
\newtheorem{ex}{EX}
\newtheorem{problem}{Problem}
\newcommand{\argmin}{\arg\!\min}
\newcommand{\argmax}{\arg\!\max}
\newcommand{\st}{\text{s.t.}}
\newcommand \dd[1]  { \,\textrm d{#1}  }

\newcounter{mytempeqncnt}

\title{Distributed Safety-Critical Control of Multi-Agent Systems with Time-Varying \\Communication Topologies} 
\author{Shiyu Cheng, \IEEEmembership{Graduate Student Member, IEEE}, Luyao Niu, \IEEEmembership{Member, IEEE}, Bhaskar Ramasubramanian, \IEEEmembership{Member, IEEE}, Andrew Clark, \IEEEmembership{Senior Member, IEEE}, and Radha Poovendran, \IEEEmembership{Fellow, IEEE}
\thanks{S. Cheng and A. Clark are with the Department of Electrical and Systems Engineering, Washington University in St. Louis, St. Louis, MO, USA. Email: \{cheng.shiyu, andrewclark\}@wustl.edu}
\thanks{B. Ramasubramanian is with the Department of Electrical and Computer Engineering, Western Washington University,
Bellingham, WA, USA. Email: ramasub@wwu.edu}
\thanks{L. Niu and R. Poovendran are with the Network Security Lab, Department of Electrical and Computer Engineering, University of Washington, Seattle, WA, USA. Email: \{luyaoniu, rp3\}@uw.edu}
}

\maketitle
\begin{abstract}
Coordinating multiple autonomous agents to reach a target region while avoiding collisions and maintaining communication connectivity is a core problem in multi-agent systems. 
In practice, agents have a limited communication range.
Thus, network links can appear and disappear as agents move, making the topology state-dependent and time-varying. 
Existing distributed solutions to multi-agent reach-avoid problems typically assume a fixed communication topology, and thus are not applicable when encountering discontinuities raised by time-varying topologies.
This paper presents a distributed optimization-based control framework that addresses these challenges through two complementary mechanisms. First, we introduce a truncation function that converts the time-varying communication graph into a smoothly state-dependent one, ensuring that constraints remain continuous as communication links are created or removed. 
Second, we employ auxiliary mismatch variables with two-time-scale dynamics to decouple globally coupled state-dependent constraints, yielding a singular perturbation system that each agent can solve using only local information and neighbor communication. 
Through singular perturbation analysis, we prove that the distributed controller guarantees collision avoidance, connectivity preservation, and convergence to the target region. 
We validate the proposed framework through numerical simulations involving multi-agent navigation with obstacles and time-varying communication topologies.
\end{abstract}


\section{Introduction}
In multi-agent systems, a central task is to coordinate multiple autonomous agents to reach a desired region while avoiding collisions with other agents and obstacles. 
This task, known as the \emph{reach-avoid problem}, arises in many applications, including cooperative robotics, autonomous vehicles, surveillance, and environmental monitoring \cite{zhang2025solving,liu2025controller,verhagen2025collaborative}.
In many of these applications, agents must not only ensure safety but also maintain connectivity of the communication network in order to enable information exchange and coordinated decision-making. Moreover, the sensing and communication capabilities of each agent are often limited by hardware constraints and environmental factors, making coordination under safety and connectivity requirements particularly challenging.

A common approach to solving reach-avoid problems is to formulate the control synthesis problem as an optimization problem.
In this framework, constraints such as collision avoidance, obstacle avoidance, and connectivity preservation can be encoded using control barrier functions (CBFs) \cite{ames2019control,xu2015robustness,capelli2020connectivity, wang2017safety}. Meanwhile, task objectives such as convergence to a desired region can be encoded using control Lyapunov functions (CLFs) \cite{sontag1983lyapunov}. By enforcing both CLF and CBF conditions within an optimization-based controller, these methods generate control inputs that drive the system toward the desired objective while guaranteeing safety constraints \cite{ames2016control,ames2014control,frauenfelder2023decentralized}.

Optimization-based controllers have been shown to provide strong guarantees of safety and task performance. However, such approaches usually rely on global state information or centralized coordination, which limits scalability as the number of agents increases.
To address these limitations, distributed implementations of optimization-based controllers have been proposed, in which each agent computes its control input using locally available information and communication with neighboring agents \cite{wang2017safety, gao2023online, wang2023distributed}.  
In these works, coupled constraints involving multiple agents are decomposed. Each agent enforces a local portion of the constraint using locally available information and information from its neighbors. Such distributed approaches improve scalability and make optimization-based control methods more suitable for large-scale multi-agent systems. However, this decomposition may introduce conservativeness, since the local constraints can be more restrictive and thus yield conservative control inputs.

Recent works in distributed optimization have addressed the challenge of constraint-coupled optimization problems by introducing auxiliary variables that decouple globally coupled constraints while preserving the feasible set of the original problem \cite{tan2021distributed,tan2025continuous, liu2026achieving, mestres2024distributed}.
Augmenting the optimization problem with these auxiliary variables allows coupled constraints to be reformulated as locally enforceable constraints that depend only on local decision variables and information exchanged with neighboring agents.
This reformulation enables distributed solution methods in which each agent solves a local optimization problem while coordinating with its neighbors to ensure global feasibility.
However, these methods are proposed based on the assumption that the inter-agent communication topology is fixed, and hence are not applicable to multi-agent systems with time-varying state-dependent communication topologies.
Designing distributed controllers that remain well-defined and stable under such dynamically evolving and state-dependent communication networks remains an open problem.

In this paper, we study the problem of steering a group of agents to a target region under a time-varying, state-dependent communication network, while maintaining safety and network connectivity.
To enable distributed control, we adopt a two-time-scale dynamical system to evolve mismatch variables that decouple coupled constraints.
To handle the addition and removal of communication links, we introduce a truncation function that enables smooth transitions. This design helps maintain well-defined mismatch variables and avoids discontinuities induced by changes in the communication topology.
The main contributions of this paper are summarized as follows.
\begin{itemize}
    \item We design a controller to steer a multi-agent system to a target region while ensuring collision avoidance and connectivity preservation under a state-dependent and time-varying communication topology.
    \item We develop a distributed optimization-based control framework that employs auxiliary mismatch variables with two-time-scale dynamics, resulting in a singular perturbation system.
    \item We introduce a truncation function to handle the appearance and disappearance of communication links, ensuring the consistency and well-definedness of mismatch variables under topology changes.
    \item We provide theoretical guarantees of convergence to the target region while maintaining safety and connectivity of the communication graph through the analysis of the associated singular perturbation system.
    \item We validate the proposed distributed control framework through numerical simulations on a multi-agent reach-avoid scenario with time-varying communication topologies, demonstrating its effectiveness.
\end{itemize}
The rest of this paper is organized as follows. Section \ref{section:related-work} reviews related work. Section \ref{section:preliminaries} introduces the preliminaries and system model. Section \ref{section:problem-formulation} presents the problem formulation and a centralized controller. Section \ref{section:distributed-algo} develops the distributed algorithm. Section \ref{section:simulation} provides simulation and experimental results. Section \ref{section:conclusion} concludes the paper.
\section{Related Work}
\label{section:related-work}
Distributed coordination of multi-agent systems has been extensively studied due to its wide range of applications in robotics, autonomous vehicles, and networked control systems \cite{oh2015survey}. In many practical scenarios, agents are required to accomplish task objectives while ensuring collision avoidance and communication connectivity preservation.
Classical approaches to multi-agent coordination include consensus-based and graph-theoretic methods, where agents achieve agreement or formation objectives through local interactions over a communication graph \cite{fax2004information, ren2004coordination}. In addition, artificial potential field methods have been widely used to enable decentralized navigation and collision avoidance, although they generally lack formal safety guarantees \cite{980728}. To address these limitations, recent works have incorporated safety requirements, and in some cases connectivity requirements, as constraints within optimization-based control frameworks.
Specifically, CLFs have been used to ensure convergence to a target region \cite{frauenfelder2023decentralized}, while CBFs provide a systematic way to encode safety constraints for both collision avoidance \cite{ames2019control,xu2015robustness, wang2017safety} and communication connectivity preservation \cite{capelli2020connectivity, de2024distributed}. 

In parallel, a large body of work has focused on developing distributed control and optimization algorithms that rely only on local interactions among neighboring agents, enabling scalable coordination while preserving the decentralized structure of the system \cite{wang2017safety, gao2023online, wang2023distributed,mestres2024distributed,tan2025continuous}. These approaches allow agents to compute control inputs using local information and communication with nearby agents while preserving safety and task guarantees.
For example, \cite{wang2017safety} proposed a decentralized implementation in which a centralized quadratic program (QP) was decomposed into a set of local QPs that could be efficiently solved by individual agents while meeting constraints of the centralized problem.
However, this approach does not in general provide optimality guarantees with respect to the original centralized optimization problem. 

More recent works \cite{tan2021distributed,tan2025continuous, liu2026achieving, mestres2023distributed, mestres2024distributed, cherukuri2017role} have developed distributed algorithms to solve constrained optimization problems while ensuring that constraints remain satisfied throughout execution of the algorithm. These approaches typically introduce mismatch or auxiliary variables to decouple any coupled constraints among agents and construct local optimization problems that can be solved independently. 
Solving the resulting separable distributed optimization problem was then shown to be equivalent to solving the original centralized optimization problem. 
In \cite{mestres2024distributed}, this approach is implemented in a feedback loop with the plant.    
A mismatch-variable formulation is used to decompose the coupled constraints, and a two-time-scale dynamical system is proposed to update the mismatch variables. 

The distributed framework adopted in our paper builds upon this class of methods. However, most existing results focus on scenarios with fixed communication topologies \cite{tan2025continuous, mestres2024distributed}. In contrast, we consider a different problem setting where the communication graph depends on the relative positions of agents and is therefore time-varying. 
Under state-dependent communication graphs, existing mismatch-variable-based frameworks encounter additional challenges. In particular, when communication links are created or removed, the mismatch-variable dynamics may become discontinuous, and the dimension of locally available variables may change over time. These issues complicate the implementation of distributed optimization dynamics.

\section{Preliminaries}
\label{section:preliminaries}
\subsection{Notations}
We use $\mathbb{Z}_{+}$, $\mathbb{R}_{+}$, and $\mathbb{R}_{\geq 0}$ to denote the sets of positive integers, positive real numbers, and nonnegative real numbers, respectively. We define $[N] = \{1, \ldots, N\}$ with $N\in \mathbb{Z}_{+}$. We use $\mathfrak{B}(\mathbb{R}^{n})$ to denote the collection of all subsets of $\mathbb{R}^{n}$. Let $\mathcal{S}_{++}^{d}$ denote the set of $d\times d$ symmetric positive definite matrices where $d\in \mathbb{Z}_{+}$. 
We use $f(x; y)$ to define a function of $x$, parameterized by $y$. For a real-valued function $f: \mathbb{R}^n\times \mathbb{R}^m\rightarrow \mathbb{R}, (x, y)\mapsto f(x,y)$, we use $\nabla_{x}f$ and $\nabla_{y}f$ to denote the column vector of partial derivatives of $f$ with respect to $x$ and $y$, respectively. For a function $g: \mathbb{R}\rightarrow \mathbb{R}^n$, the notation $\lim_{t\uparrow t^{\prime}}g(t)$ denotes the left-hand limits of $g$ at $t^{\prime}$, i.e., $\lim_{t\uparrow t^{\prime}}g(t) = \lim_{t\rightarrow (t^{\prime})^{-}}g(t)$.
For vector $v$, we use $v^{(i)}$ to denote the $i$-th entry in $v$.
For vectors $v_1, \ldots, v_k$, we define the stacking operator
$\text{col}(v_1, \ldots, v_k) \triangleq
[(v_1)^T, \ldots, (v_k)^T]^T$. For an indexed family
$\{v_\alpha\}_{\alpha \in \mathcal{S}}$ over a finite
ordered set $\mathcal{S}$, we write
$\text{col}(\{v_\alpha\}_{\alpha \in \mathcal{S}})$ to
denote stacking according to that order.
Let $I^m$ denote the $m\times m$ identity matrix, $0^{m\times n}$ denote the zero matrix of dimensions $m\times n$, $\mathbf{1}^m$ denote the $m$- dimensional vector of ones, and $\mathbf{0}^m$ denote the $m$-dimensional vector of zeros. 
For $x\in \mathbb{R}^n$, we use $\|x\|$ to denote its Euclidean norm. We use $B_r$ to denote the open ball of radius $r$ centered at the origin, and $\overline{B}_r$ to denote the corresponding closed ball. 
For a set $C\subseteq \mathbb{R}^n$ and ${x}\in \mathbb{R}^n$, we use $d_{C}({x}) \triangleq \inf \{\|{x}-{y}\|:{y}\in C\}$ to represent the distance between ${x}$ and set $C$. We define $C+\overline{B}_r = \{x\in \mathbb{R}^n: d_{C}(x)\leq r\}$. The notation $|C|$ is used to denote the cardinality of the finite set $C$. We use $\text{Proj}_{n}$ to denote the projection onto $\mathbb{R}^{n}$.
For $a \in \mathbb{R}$ and $b\in \mathbb{R}_{+}$, we set 
\begin{align*}
    [a]_b^{+} = \left\{
    \begin{array}{ll}
    a, & b>0\\
    \max\{0, a\}, & b=0 
    \end{array}\right.
\end{align*}
For $\mathbf{a} \in \mathbb{R}^n$ and $\mathbf{b}\in \mathbb{R}_{+}^n$, $[\mathbf{a}]_\mathbf{b}^{+}\in \mathbb{R}^n$ denotes the vector obtained by applying $[a^{(i)}]_{b^{(i)}}^{+}$ component-wise for $i\in [n]$.

\subsection{System Model}
Consider a multi-agent system with $N$ agents. We use $x_i(t) = \text{col}(p_i(t), q_i(t))  \in \mathcal{X}\subset \mathbb{R}^{n}$ to denote the state of agent $i$ at time $t$, where $p_i(t)\in \mathbb{R}^l$ denotes the position of agent $i$ and $q_i(t)\in \mathbb{R}^{n-l}$ denotes the other states, such as velocity and acceleration, with $l\in \{2, 3\}$ and $n-l\geq 0$.
Each agent $i$ follows the control-affine system dynamics
\begin{align}
\label{eq:control-affine}
    \dot{x}_i(t) = F_i(x_i(t))+G_i(x_i(t))u_i(t),
\end{align}
where $F_i: \mathbb{R}^{n}\rightarrow \mathbb{R}^{n}$ and $G_i: \mathbb{R}^{n}\rightarrow \mathbb{R}^{n\times m}$ are locally Lipschitz functions, with $u_i\in \mathcal{U}\subset \mathbb{R}^{m}$. We denote $\mathbf{x}(t) = \text{col}(x_1(t), \ldots, x_N(t))$ and $\mathbf{u}(t) = \text{col}(u_1(t), \ldots, u_N(t))$.
Let $u_i^{\ast}(\mathbf{x}(t))$ be a feedback controller such that the dynamical system
\begin{align}
\label{eq:feedback-control}
    \dot{x}_i(t) = F_i(x_i(t))+G_i(x_i(t))u_i^{\ast}(\mathbf{x}(t))
\end{align}
is locally Lipschitz. For any initial condition $x_i(t_0)\in \mathcal{X}$, there exists a maximum interval of existence $I(x_i(t_0)) = [t_0, T_{max})$, such that $x_i(t)$ is the unique solution to \eqref{eq:feedback-control} on $I(x_i(t_0))$. If the system \eqref{eq:feedback-control} is forward complete \cite{khalil2015nonlinear}, then $T_{\max} = \infty$.

We assume that two agents can exchange information when the distance between their positions is less than $d_c\in\mathbb{R}_{+}$.
We use an undirected graph $\mathcal{G}(t) = (\mathcal{V}, \mathcal{E}(\mathbf{x}(t)))$ to denote the communication graph of the system, where $\mathcal{V} = [N]$ denotes the set of agents in the system, and $\mathcal{E}(\mathbf{x}(t)) \triangleq \{(i,j)\in \mathcal{V}\times \mathcal{V}: \|p_i(t)-p_j(t)\|< d_c, i\neq j\}$ denotes the set of communication edges at time $t$. We define the set of neighbors of agent $i$ at time $t$ as $\mathcal{N}_i(\mathbf{x}(t)) \triangleq \{j\in \mathcal{V}: \|p_i(t)-p_j(t)\|< d_c, j\neq i\}$. 

\subsection{Background on Control Barrier Functions}
We present the background of control barrier functions (CBFs) as follows \cite{ames2019control}. 
Consider a set $\mathcal{C}$ defined as the superlevel set of a continuously differentiable function $h:\mathcal{X}\subset \mathbb{R}^n\rightarrow \mathbb{R}$ as 
\begin{align}
\label{eq:safe-set}
    \mathcal{C} &= \{x\in \mathcal{X}\subset\mathbb{R}^n: h(x)\geq0\}\\
    \partial \mathcal{C} &= \{x\in \mathcal{X}\subset \mathbb{R}^n: h(x)= 0\} 
    \nonumber
\end{align}

\begin{definition}[Safety]
    The set $\mathcal{C}$ is forward invariant under feedback control law \eqref{eq:feedback-control} if for every $x(t_0)\in \mathcal{C}$, $x(t)\in \mathcal{C}$ for all $t\in I(x(t_0))$. The system \eqref{eq:feedback-control} is safe with respect to the set $\mathcal{C}$ if $\mathcal{C}$ is forward invariant.
\end{definition}

    A function $\tilde{f}: \mathbb{R}\rightarrow\mathbb{R}$ is in extended class $\mathcal{K}_{\infty}$ if $\tilde{f}$ is strictly increasing,  $\tilde{f}(0)=0$, and $\lim_{t\rightarrow\infty}\tilde{f}(t)=\infty$. Next, we will introduce the definition of CBFs.

\begin{definition}[Control Barrier Function \cite{ames2019control}]
    Let $\mathcal{C}\subset \mathbb{R}^n$ be defined as in \eqref{eq:safe-set}, then $h(x)$ is a control barrier function (CBF) if there exists an extended class $\mathcal{K}_{\infty}$ function $\alpha$ such that for the control system \eqref{eq:control-affine}:
    \begin{align}
    \label{eq:def-cbf}
        \sup_{u\in \mathcal{U}}[L_Fh(x)+L_Gh(x)u]\geq -\alpha(h(x)),
    \end{align}
for all $x\in \mathcal{X}$, where $L_Fh(x) = \frac{\partial h(x)}{\partial x}F(x)$ and $L_Gh(x) = \frac{\partial h(x)}{\partial x}G(x)$ denote Lie derivatives.
\end{definition}
The set of control inputs that satisfies \eqref{eq:def-cbf} is denoted as
$$K_{\text{cbf}}(x)=\{u\in \mathcal{U}: L_Fh(x)+L_Gh(x)u+\alpha(h(x))\geq 0\}.$$
\begin{theorem}[\cite{ames2019control}]
\label{theorem:cbf}
    Let $\mathcal{C}\subset \mathbb{R}^n$ be a set defined as the superlevel set of a continuously differentiable function $h: \mathcal{X}\subset \mathbb{R}^n\rightarrow \mathbb{R}$. If $h$ is a CBF on $\mathcal{X}$ and $\frac{\partial h}{\partial x}(x)\neq 0$ for all $x\in \partial \mathcal{C}$, then any locally Lipschitz continuous controller
    $u: \mathcal{X}\rightarrow \mathcal{U}$ satisfying $u(x)\in K_{\text{cbf}}(x)$, $\forall x\in \mathcal{X}$ for the system \eqref{eq:control-affine} renders the set $\mathcal{C}$ safe.
\end{theorem}

\subsection{Background on Singular Perturbations}
Consider a differential inclusion
\begin{align}
\label{eq:diff-inclusion}
    \begin{pmatrix}
        \dot{x}(t)\\\tau \dot{y}(t)
    \end{pmatrix}    \in J(x(t), y(t))
\end{align}
where $\tau>0$, $(x, y)\in \mathbb{R}^{n_x}\times \mathbb{R}^{n_y}$, $t\geq 0$, and $J: \mathbb{R}^{n_x}\times \mathbb{R}^{n_y}\rightarrow \mathfrak{B}(\mathbb{R}^{n_x+n_y})$.
For a given $\tau$, a solution of \eqref{eq:diff-inclusion} on an interval $\overline{I}$ is a pair $(x_{\tau}, y_{\tau})$ of absolutely continuous functions such that
$$\begin{pmatrix}
        \dot{x}_{\tau}(t)\\\tau \dot{y}_{\tau}(t)
    \end{pmatrix}\in J({x}_{\tau}(t),  {y}_{\tau}(t)), \text{for almost every } t\in \overline{I}.$$
We define $S_{\tau}(x_0, y_0)$ as the set of all the solutions of \eqref{eq:diff-inclusion} with $x(0)=x_0$ and $y(0)=y_0$. 
We assume that solutions to \eqref{eq:diff-inclusion} exist for all initial conditions for all time. When $\tau$ is small, we call $x$ the slow variable and $y$ the fast variable.
For a fixed $x$, 
\begin{align}
\label{eq:fast-system}
    \tau\dot{y}(t) \in \text{Proj}_{n_y}J(x, y(t)).
\end{align}
For a given $x\in \mathbb{R}^{n_x}$, suppose that there exists a Lipschitz set-valued map $\mathcal{H}: \mathbb{R}^{n_x}\rightarrow \mathfrak{B}(\mathbb{R}^{n_y})$, with compact convex values, such that the set $\mathcal{H}(x)$ is asymptotically stable\footnote{The set $\mathcal{H}(x)$ is asymptotically stable if  $\overline{\delta}>0$ can be chosen such that $\|y(0)\|<\overline{\delta}$ implies $\lim_{t\rightarrow \infty}d_{\mathcal{H}(x)}(y(t))=0$.} with respect to the system \eqref{eq:fast-system}. We define 
\begin{align}
    \label{eq:slow-system}
    D(x) = \overline{co}(\text{Proj}_{n_x}J(x, \mathcal{H}(x)))
\end{align}
where $\overline{co}(X)$ denotes the closed convex hull of set $X$. The set-valued map $D(x)$ defines the reduced slow system, given by the differential inclusion $\dot{x}\in D(x)$.
Let $S_D(x_0)$ denote the set of solutions of the system  $\dot{x}\in D(x)$ starting at $x_0$.
We define uniform global asymptotic stability (UGAS) as follows.
\begin{definition}[UGAS \cite{watbled2005singular}]
    \label{def:ugas}
    A set $A\subseteq \mathbb{R}^{n_x}$ is UGAS for the 
    differential inclusion $\dot{x} \in D(x)$ if there exist a nondecreasing  
    function $\psi: \mathbb{R}_{+}\rightarrow\mathbb{R}_{+}$ with $\lim_{r\rightarrow 0}\psi(r)=0$ and a function $\mathcal{T}: \mathbb{R}_{+}\times 
    \mathbb{R}_{+}\rightarrow\mathbb{R}_{+}$ such 
    that for any $R>0$, any $\nu>0$, and any solution 
    $x(\cdot)$ with $x(0)\in A+\overline{B}_R$, the 
    following hold:
    \begin{enumerate}
        \item $\lim_{t\rightarrow \infty}d_{A}(x(t))=0$,
        \item $d_{A}(x(t))\leq \psi(R)$ for all $t\geq 0$,
        \item $d_{A}(x(t))\leq \nu$ for all 
              $t\geq \mathcal{T}(R, \nu)$.
    \end{enumerate}
\end{definition}
Next, we introduce the assumptions required to ensure the convergence of trajectories for a differential inclusion system. 

\textbf{HF}: The map $J:\mathbb{R}^{n_x+n_y}\rightarrow \mathfrak{B}(\mathbb{R}^{n_x+n_y})$ has nonempty convex compact values, and is L-Lipschitz on $\mathbb{R}^{n_x}\times \mathbb{R}^{n_y}$, i.e., there exists $L>0$ such that, for any $(x,y)\in \mathbb{R}^{n_x}\times \mathbb{R}^{n_y}$, we have 
\begin{align*}
    \sup_{(u,v)\in J(x,y)} \{\|u\|+\|v\|\}\leq L(\|x\|+\|y\|+1).
\end{align*}

\textbf{H1}: The map $\mathcal{H}: \mathbb{R}^{n_x}\rightarrow \mathfrak{B}(\mathbb{R}^{n_y})$ is Lipschitz on $\mathbb{R}^{n_x}$ with a Lipschitz constant $L_K$, i.e., 
    $\inf_{u_1\in \mathcal{H}(x_1), u_2\in \mathcal{H}(x_2)} \|u_1-u_2\|\leq L_K \|x_1-x_2\|$.
For every $x\in \mathbb{R}^{n_x}$, the set $\mathcal{H}(x)$ is a compact convex subset of $\mathbb{R}^{n_y}$.

\textbf{H2}: There exists a nonempty compact subset $A$ of $\mathbb{R}^{n_x}$ which is UGAS for the differential inclusion $\dot{x}\in D(x)$, with $\psi$ and $T$ as defined in Definition \ref{def:ugas}.

\textbf{H3}: Every set $\mathcal{H}(x)$ is UGAS for the differential inclusion \eqref{eq:diff-inclusion}, with functions $\psi_R$ and $T_R$ as in Definition \ref{def:ugas}, being such that: $\forall R>0$, $\forall x\in \overline{B}_R$, $\forall r>0$, $\forall y\in \mathcal{H}(x)+\overline{B}_r$, $\forall \nu >0$, every trajectory $y(\cdot)\in S_x(y)$ satisfies:
\begin{enumerate}
    \item $\lim_{t\rightarrow + \infty}d_{\mathcal{H}(x)}(y(t))= 0$
    \item $d_{\mathcal{H}(x)}(y(t))\leq \psi_R(r), \ \forall t\geq 0$
    \item $d_{\mathcal{H}(x)}(y(t))\leq \nu, \ \forall t\geq T_R(r, \nu)$
\end{enumerate}
where $S_x(y)$ denotes the set of all the solutions with initial condition $y(0)=y$, i.e., 
$S_x(y) = \{y(\cdot):\tau \dot{y}(s)\in \text{Proj}_{n_y}J(x, y(s)) \text{ almost everywhere, } y(0)=y\}$.

Next, we present results on the convergence of trajectories of \eqref{eq:diff-inclusion}.
\begin{theorem}[\cite{watbled2005singular}]
\label{:theorem-converge}
    Let assumptions \textbf{HF}, \textbf{H1}, \textbf{H2}, and \textbf{H3} hold. Then, for any $x\in \mathbb{R}^{n_x}$, there exists $r>0$ such that, for any $y\in \mathcal{H}(x)+B_r$, and for every sequence $(x_{\tau}(\cdot), y_{\tau}(\cdot))\in S_{\tau}(x, y)$, there exists a sequence  $(\overline{x}_{\tau}(\cdot))\in S_D(x)$ such that 
    \begin{align}
        \forall t_1>0, &\lim_{\tau \rightarrow 0}\sup_{t\geq t_1}d_{\mathcal{H}(x_{\tau}(t))}(y_{\tau}(t))=0 \label{eq:converge-result-1}\\
        \forall T>0, &\lim_{\tau \rightarrow 0}\sup_{t\in [0, T]}||x_{\tau}(t)-\overline{x}_{\tau}(t)||=0\label{eq:converge-result-2}
    \end{align}
Moreover, for any $\eta>0$, there exists $T_{\eta}>0$ with the following property. For each $T>T_{\eta}$, there exists $\tau_0>0$ such that for all $\tau\in (0, \tau_0)$, we have
\begin{align}
\label{eq:converge-result-3}
    \sup_{t\in [0, T]}||x_{\tau}(t)-\overline{x}_{\tau}(t)||&+\sup_{t\geq T}d_A(x_{\tau}(t))\nonumber\\
    &\quad\quad\quad+\sup_{t\geq T}d_A(\overline{x}_{\tau}(t))\leq \eta.
\end{align}
\end{theorem}
Equation~\eqref{eq:converge-result-1} implies that $\forall t_1>0$, $\exists \tau>0$ such that fast variable $y_{\tau}(t)$ converges to the set $\mathcal{H}(x_{\tau}(t))$ within any prescribed time $t_1$. 
Moreover, for sufficiently small $\tau$, the trajectory of the slow variable $x_{\tau}$ remains close to some trajectory $\overline{x}_{\tau}\in S_D(x)$ uniformly on any compact time interval, as shown in \eqref{eq:converge-result-2}. Finally, \eqref{eq:converge-result-3} implies that for any $\eta>0$ there exists $T_{\eta}$ such that for all $T\geq T_{\eta}$, by choosing $\tau$ sufficiently small, the sum of (i) the distance between $x_{\tau}(t)$ and $\overline{x}_{\tau}(t)$, (ii) the distance between $x_{\tau}(t)$ and the target set $A$, and (iii) the distance between $\overline{x}_{\tau}(t)$ and $A$ is bounded above by $\eta$.
\subsection{Background on Projected Saddle-Point Dynamics}
\label{section:background-PSD}
Let $f:\mathbb{R}^n \rightarrow \mathbb{R}$ and $g: \mathbb{R}^n \rightarrow \mathbb{R}^{p}$ be continuously differentiable functions, whose derivatives are locally Lipschitz,  and let $A\in \mathbb{R}^{m\times n}$ and $b\in \mathbb{R}^m$. For the optimization problem
\begin{align}
\label{eq:nonlinear-opt}
    \min_{\tilde{x}\in \mathbb{R}^n}&\quad f(\tilde{x})\\
        \text{subject to} 
        & \quad g(\tilde{x})\leq 0 \nonumber\\
        & \quad A\tilde{x}=b \nonumber
\end{align}
let $\mathcal{L}: \mathbb{R}^n\times \mathbb{R}^p_{\geq 0}\times \mathbb{R}^m\rightarrow \mathbb{R}$ be defined as $\mathcal{L}(\tilde{x}, \tilde{y}, \tilde{\mu})\triangleq f(\tilde{x})+\tilde{y}^Tg(\tilde{x})+\tilde{\mu}^T(A\tilde{x}-b)$. We use $\text{Saddle}(\mathcal{L})$ to denote the set of saddle points of $\mathcal{L}$. The projected saddle-point dynamics for $\mathcal{L}$ are defined as 
\begin{subequations}
\label{eq:projected-saddle-dynamics}
\begin{align}
    \dot{\tilde{x}} &= -\nabla_{\tilde{x}} \mathcal{L}(\tilde{x}, \tilde{y}, \tilde{\mu}),\\
    \dot{\tilde{y}} &= [\nabla_{\tilde{y}}\mathcal{L}(\tilde{x}, \tilde{y}, \tilde{\mu})]_{\tilde{y}}^{+},\\\dot{\tilde{\mu}} &= \nabla_{\tilde{\mu}} \mathcal{L}(\tilde{x}, \tilde{y}, \tilde{\mu}).
\end{align}
\end{subequations}
If $f(\tilde{x})$ is strictly convex and $g(\tilde{x})$ is convex, then $\mathcal{L}$ has a unique saddle point, which corresponds to the KKT point of \eqref{eq:nonlinear-opt}.  Furthermore, this point is globally asymptotically stable under the dynamics \eqref{eq:projected-saddle-dynamics}. The following theorem, restated from \cite[Theorem~5.1]{cherukuri2017role}, formalizes this result.
\begin{theorem}[{\cite{cherukuri2017role}}]
\label{theorem:5-1}
    Suppose that $f(\tilde{x})$ is strictly convex and twice continuously differentiable, and that $g(\tilde{x})$ is convex and twice continuously differentiable. For each $(\tilde{x}, \tilde{y}, \tilde{\mu})\in \mathbb{R}^n\times \mathbb{R}^p_{\geq 0}\times \mathbb{R}^m$, define the index set of active constraints
    \begin{align}
        \mathcal{J}(\tilde{x}, \tilde{y}, \tilde{\mu}) \triangleq \{j\in \{1, \ldots, p\}: \tilde{y}_j = 0, (\nabla_{\tilde{y}}\mathcal{L}(\tilde{x}, \tilde{y}, \tilde{\mu}))^{(j)}<0\}.
    \end{align}
Then, the function $W: \mathbb{R}^n\times \mathbb{R}^p_{\geq 0}\times \mathbb{R}^m\rightarrow \mathbb{R}$, 
\begin{align}
    W(\tilde{x}, \tilde{y}, \tilde{\mu})&\triangleq \frac{1}{2}\big(\|\nabla_{\tilde{x}}\mathcal{L}(\tilde{x}, \tilde{y}, \tilde{\mu})\|^2 + \|\nabla_{\tilde{\mu}}\mathcal{L}(\tilde{x}, \tilde{y}, \tilde{\mu})\|^2\nonumber\\
    & + \sum_{j\in \{1, \ldots, p\}\setminus\mathcal{J}(\tilde{x}, \tilde{y}, \tilde{\mu})}((\nabla_{\tilde{y}}\mathcal{L}(\tilde{x}, \tilde{y}, \tilde{\mu}))^{(j)})^2)\big)\nonumber
\end{align}
is nonnegative everywhere in its domain and $ W(\tilde{x}, \tilde{y}, \tilde{\mu}) = 0$ if and only if $(\tilde{x}, \tilde{y}, \tilde{\mu})\in \text{Saddle}(\mathcal{L})$. For any trajectory $(\tilde{x}(t), \tilde{y}(t), \tilde{\mu}(t))$ of \eqref{eq:projected-saddle-dynamics}, the function $t\mapsto W(\tilde{x}(t), \tilde{y}(t), \tilde{\mu}(t))$
\begin{enumerate}
    \item is differentiable almost everywhere and if $(\tilde{x}, \tilde{y}, \tilde{\mu})\notin \text{Saddle}(\mathcal{L})$ for some $t\geq 0$, then $\frac{d}{dt}W(\tilde{x}(t), \tilde{y}(t), \tilde{\mu}(t))<0$ provided the derivative exists. Moreover, for any sequence of times $\{t_k\}_{k=1}^{\infty}$ such that $t_k \rightarrow t$ and $\frac{d}{dt}W(\tilde{x}(t_k), \tilde{y}(t_k), \tilde{\mu}(t_k))$ exists for every $t_k$, we have $\lim \sup_{k\rightarrow \infty}\frac{d}{dt}W(\tilde{x}(t_k), \tilde{y}(t_k), \tilde{\mu}(t_k))<0$;
    \item is right continuous and at any point of discontinuity $t^{\prime}\geq0$, we have $W(\tilde{x}(t^{\prime}), \tilde{y}(t^{\prime}), \tilde{\mu}(t^{\prime}))\leq \lim_{t\uparrow t^{\prime}}W(\tilde{x}(t), \tilde{y}(t), \tilde{\mu}(t))$.
\end{enumerate}
As a consequence, $\text{Saddle}(\mathcal{L})$ is globally asymptotically stable under the dynamical system \eqref{eq:nonlinear-opt} and the convergence of trajectories is to a point.
\end{theorem}

\section{Problem Formulation}
\label{section:problem-formulation}
In this work, we aim to design controllers for all agents to steer them from their initial positions to the target region, while ensuring inter-agent collision avoidance,  obstacle avoidance, and the connectivity of the communication graph. We define these properties formally as follows:
\begin{enumerate}         
    \item \textit{Inter-agent collision avoidance}: Let $d_0\in \mathbb{R}_{+}$ denote the minimum safe distance with $0<d_0<d_c$. The inter-agent collision avoidance is to ensure $\|p_i(t)-p_j(t)\|\geq d_0$ for all $t\geq t_0$, $\forall i\neq j$.
    \item \textit{Obstacle avoidance}: We assume there are $K\in \mathbb{Z}_{+}$ obstacles in the environment, whose positions are denoted as $s_k\in \mathbb{R}^l$ for all $k\in [K]$. We set $\|p_i(t)-s_k\|\geq r_k$ with $r_k\in \mathbb{R}_{+}$ for all $k\in [K]$ as a sufficient condition ensuring obstacle avoidance for agent $i$. 
    \item \textit{Connectivity of the communication graph}: The graph $\mathcal{G}(t)$ is connected at $t$ if there exists a path between every pair of vertices.
    \item \textit{Convergence to the target region}: The target region for each agent $i$ is given by $\Omega_i \triangleq \{x_i\in \mathbb{R}^n:(x_i-\overline{x})^T\Sigma(x_i-\overline{x})\leq 1 \}$ where $\Sigma\in \mathcal{S}_{++}^{n}$. Without loss of generality, we set $\overline{x} = \mathbf{0}^n$, as the ellipsoid $\Omega_i$ can always be recentered at the origin through a translation of coordinates, which does not affect the analysis.
\end{enumerate}
We define the safety constraint at time $t$ as  
\begin{align}
    \label{eq:safety-guarantee}
    \left\{
    \begin{array}{ll}
    \|p_i(t)-p_j(t)\|\geq d_0,     &\forall i\neq j  \\
    \|p_i(t)-s_k\|\geq r_k,    & \forall i\in \mathcal{V}, \forall k\in [K]
    \end{array}\right.
\end{align}

\begin{problem}
\label{problem:statement}
Given initial states  $\mathbf{x}(t_0)$ satisfying \eqref{eq:safety-guarantee} and a connected $\mathcal{G}(t_0)$, compute $\{\mathbf{u}(t): t\geq t_0\}$ such that there exists $T\geq t_0$ with $x_i(t)\in \Omega_i$
, $\forall i\in \mathcal{V}$, 
$\forall t\geq T$, while ensuring  \eqref{eq:safety-guarantee}   holds for all $t\geq t_0$. 
\end{problem}

Next, we construct sufficient conditions to ensure inter-agent collision avoidance, obstacle avoidance, connectivity of the communication graph, and the convergence of agents to the target region, respectively. We drop the argument $t$ when the context is clear. 

\textit{Inter-agent collision avoidance}: In order to ensure inter-agent collision avoidance between agents $i$ and $j$, we need $\|p_i(t)-p_j(t)\|\geq d_0$ for all $t\geq t_0$. Since agents separated by a sufficiently large distance cannot collide over a finite time horizon, we only consider collision avoidance between agent pairs whose relative distance is below the communication range $d_c$. 
We assume that there exists a CBF $h_{\text{inter}}^{i,j}:\mathcal{X}\times \mathcal{X}\rightarrow\mathbb{R}$, $(x_i,x_j) \mapsto h_{\text{inter}}^{i,j}(x_i,x_j)$, such that 
\begin{align*}
    \{(x_i, x_j): h_{\text{inter}}^{i,j}(x_i, x_j)\geq 0\}\subseteq \{(x_i, x_j): \|p_i-p_j\|\geq d_0\}. 
\end{align*}
We define 
\begin{align*}
    &g_{\text{inter}}^{i,j}(u_i, u_j; x_i, x_j) \triangleq L_{F_i}h_{\text{inter}}^{i,j}(x_i, x_j)+L_{G_i}h_{\text{inter}}^{i,j}(x_i, x_j)u_i\\
    &+ L_{F_j}h_{\text{inter}}^{i,j}(x_i, x_j)+L_{G_j}h_{\text{inter}}^{i,j}(x_i, x_j)u_j + \gamma_1(h_{\text{inter}}^{i,j}(x_i, x_j)),
\end{align*}
where $\gamma_1(\cdot)$ is an extended class $\mathcal{K}_{\infty}$ function.
The corresponding CBF condition can be represented as $-g_{\text{inter}}^{i,j}(u_i, u_j; x_i, x_j)\leq 0$.

\textit{Obstacle avoidance}: We assume there exists a CBF $h_{\text{obs}}^{i,k}: \mathcal{X}\rightarrow \mathbb{R}$, $x_i \mapsto h_{\text{obs}}^{i,k}(x_i)$,  such that 
\begin{align*}
    \{x_i: h_{\text{obs}}^{i,k}(x_i)\geq 0\}\subseteq \{x_i: \|p_i-s_k\|\geq r_k\}. 
\end{align*}
We define 
\begin{align*}
    g_{\text{obs}}^{i,k}(u_i; x_i) \triangleq L_{F_i} h_{\text{obs}}^{i,k}(x_i)+L_{G_i} h_{\text{obs}}^{i,k}(x_i)u_i + \gamma_2  (h_{\text{obs}}^{i,k}(x_i)).
\end{align*}
where $\gamma_2(\cdot)$ is an extended class $\mathcal{K}_{\infty}$ function.
The corresponding CBF condition is $-g_{\text{obs}}^{i,k}(u_i; x_i) \leq 0$.

For particular classes of system dynamics, existing works \cite{wang2016safety, wang2017safety,mestres2024distributed} have developed feasible inter-agent and obstacle collision CBFs that can be interpreted as special cases of the control-affine systems considered in this work.

\textit{Connectivity of the communication graph}: 
The following sufficient condition is used to ensure the connectivity of $\mathcal{G}(t)$ for all $t\geq t_0$. We introduce a communication margin $0<\epsilon_c<d_c$ and construct the Laplacian matrix using only those edges $(i,j)$ satisfying $\|p_i(t)-p_j(t)\|\leq d_c-\epsilon_c$. 
Let $\lambda_2(\mathbf{x}(t);d_c-\epsilon_c)$ denote the second-smallest eigenvalue of the Laplacian matrix at time $t$. If $\lambda_2(\mathbf{x}(t); d_c-\epsilon_c)>0$, then we have that the undirected graph $\mathcal{G}(t)$ is connected. Some works use CBF to ensure the connectivity of multi-agent systems. Following \cite{capelli2020connectivity, de2024distributed}, we assume there is a CBF $h_c(\mathbf{x}): \mathbb{R}^{nN}\rightarrow\mathbb{R}$,  $\mathbf{x} \mapsto h_c(\mathbf{x})$,  such that
\begin{align*}
    \{\mathbf{x}: h_c(\mathbf{x})\geq 0\}\subseteq \{\mathbf{x}: \lambda_2(\mathbf{x};d_c-\epsilon_c)>0\}.
\end{align*}
Assume that the initial condition satisfies $h_c(\mathbf{x}(t_0))\geq 0$. We define $$g^c(\mathbf{x}, \mathbf{u}) \triangleq \sum_{i=1}^{N}(L_{F_i}h_c(\mathbf{x})+L_{G_i}h_c(\mathbf{x})u_i)+\gamma_4\cdot h_c(\mathbf{x}),$$ where $\gamma_4>0$. The corresponding CBF condition is $-g^c(\mathbf{x}, \mathbf{u})\leq 0$.

\textit{Convergence to the target region}: We construct a sufficient condition to ensure each agent $i$ converges to $\Omega_i$ and stays within it. 
Let $0<\epsilon<1$, and define the corresponding subset of $\Omega_i$ as $\Omega_i^{\epsilon} \triangleq \{x_i\in \mathbb{R}^n: x_i^T\Sigma x_i-1+\epsilon\leq 0\}$. The task is to construct a controller such that 
there exists $T\geq t_0$, $x_i(t)\in \Omega_i^{\epsilon}\subset\Omega_i$, $\forall t\geq T$. To achieve this goal, we impose a control Lyapunov function (CLF)-like condition only on agents $i$ satisfying $x_i(t)\notin \Omega_i^{\epsilon}$, so as to steer them toward the $\Omega_i^{\epsilon}$. Assume there exists a positive semi-definite function $V_i: \mathcal{X}\rightarrow \mathbb{R}_{\geq0}$, $x_i \mapsto V_i(x_i)$, which is locally Lipchitz and twice continuously differentiable in $x_i$ such that 
\begin{align*}
    V_i(x_i)&= 0, \ \forall x_i\in \Omega_i^{\epsilon} \\
   \beta_1 \cdot d_{\Omega_i^{\epsilon}}(x_i)&\leq  V_i(x_i)\leq \beta_2 \cdot d_{\Omega_i^{\epsilon}}(x_i)
\end{align*}
where $\beta_2>\beta_1>0$.
A candidate CLF-like function is defined as $V(\mathbf{x}) \triangleq \sum_{i = 1}^{N}h_i(x_i)\cdot V_i(x_i)$, where 
\begin{align*}
    h_i(x_i) \triangleq \left\{
    \begin{array}{ll}
    e^{-\frac{1}{V_i(x_i)}}, & x_i\notin \Omega_i^{\epsilon}\\
    0, & x_i\in \Omega_i^{\epsilon}
    \end{array}\right. 
\end{align*}
The function $h_i(x_i)$ is smooth in $x_i$, which ensures that $V(\mathbf{x})$ is smooth. Moreover, $V(\mathbf{x})$ is enforced only for agents $i$ with $x_i\notin \Omega_i^{\epsilon}$.
The corresponding convergence constraint can be represented as
\begin{align}
\label{eq:clf-condition}
    \dot{V}(\mathbf{x}) + \gamma_3 V(\mathbf{x})\leq 0,
\end{align}
where $\gamma_3>0$.
Based on the fact that $h_i(x_i)$, $\frac{\partial h_i(x_i)}{\partial x_i}$, $V_i(x_i)$, and $\frac{\partial V_i(x_i)}{\partial x_i}$ are continuous in $x_i$, we have that $\dot{V}(\mathbf{x}) + \gamma_3 V(\mathbf{x})$ is continuous in $\mathbf{x}$.
Define
\begin{align}
   \label{eq:clf} 
   &g_{clf}(\mathbf{u};\mathbf{x}) \triangleq \sum_{i=1}^{N}\big[L_{F_i}(h_i(x_i)V_i(x_i))\nonumber\\
   &\quad\quad\quad+L_{G_i}(h_i(x_i)V_i(x_i))u_i + \gamma_3 h_i(x_i) V_i(x_i)\big]
\end{align}
and recast \eqref{eq:clf-condition} as $g_{clf}(\mathbf{u};\mathbf{x}) \leq 0$.

Finally, considering the restriction on the actuator, we set $\|u_i\|_{\infty}\leq c, \forall i\in \mathcal{V}$
with $c>0$. 
Computing $\mathbf{u}(t)$ at time $t$ to ensure inter-agent collision avoidance, obstacle avoidance, connectivity of the communication graph, and the convergence of agents to the target region can be formulated as the following optimization problem
\begin{subequations}\label{eq:problem-formulation}
\begin{align}    \min_{\mathbf{u}}&\quad \frac{1}{2}\|\mathbf{u}\|^2\\
        \text{s.t.} 
        & \quad -g_{\text{inter}}^{i,j}(u_i, u_j; x_i, x_j) \leq 0, \forall (i,j)\in \mathcal{E}(\mathbf{x})\label{eq:opt-inter}\\
        &\quad -g_{\text{obs}}^{i,k}(u_i; x_i)\leq 0, \forall i\in \mathcal{V}, k\in [K] \label{eq:16-obs}\\
        & \quad -g^c(\mathbf{u}; \mathbf{x} )\leq 0\label{eq:opt-connectivity} \\
        &\quad g_{clf}(\mathbf{u};\mathbf{x}) \leq 0\label{eq:opt-clf}\\
        &\quad \|u_i\|_{\infty}\leq c, \forall i\in \mathcal{V}
\end{align}
\end{subequations}
\begin{assumption}
\label{assumption:feasible}
    Problem \eqref{eq:problem-formulation} is feasible for all $\mathbf{x}\in \mathcal{X}^{N}$.
\end{assumption}

This assumption is standard in the distributed 
CBF-based control literature; see, e.g., 
\cite{mestres2024distributed}. 
Techniques for 
verification and synthesis of compatible CLF and CBFs can be found in \cite{dai2024verification}.
Define $A = \prod_{i=1}^{N}\Omega_i$ and $A^{\epsilon} = \prod_{i=1}^{N}\Omega_i^{\epsilon}$.  In what follows, we prove that the controller $\mathbf{u}(t)$ obtained by solving \eqref{eq:problem-formulation} is safe and stable. 
\begin{theorem}
    \label{theorem:clf-convergence}
    Under Assumption~\ref{assumption:feasible}, if $\mathbf{u}(t)$ is the solution to \eqref{eq:problem-formulation} for all $t\geq t_0$, then the controller is safe and stable while preserving connectivity:
    \begin{enumerate}
        \item $\|p_i(t)-p_j(t)\|\geq d_0$,  $\forall i\neq j$, and $\|p_i(t)-s_k\|\geq r_k$  $\forall i\in \mathcal{V}$, $k\in [K]$, for all $t\geq t_0$.
        \item The graph $\mathcal{G}(t)$ is connected for all $t\geq t_0$.
        \item There exists $T\geq t_0$ such that $\mathbf{x}(t)\in A$ for all $t\geq T$. Furthermore, set $A^{\epsilon}$ is UGAS for \eqref{eq:control-affine}.
    \end{enumerate}
\end{theorem}
The proof of Theorem \ref{theorem:clf-convergence} is included in the Appendix~\ref{appendix:theorem3}.

This section presents a centralized approach to synthesize the controllers of agents via solving \eqref{eq:problem-formulation}. 
This framework relies on the availability of global information $(\mathbf{u}, \mathbf{x})$ and a central coordinator. In the following section, we present a distributed algorithm that achieves safety and stability by using only local measurements and neighbor-based communication.

\section{Distributed Algorithm}
\label{section:distributed-algo}
Some constraints in the optimization problem \eqref{eq:problem-formulation} involve the states and control inputs of multiple agents. 
To derive a distributed algorithm, we introduce constraint-mismatch variables following \cite{mestres2024distributed, tan2025continuous}. As a first step, we reformulate \eqref{eq:problem-formulation} as an equivalent centralized problem that incorporates these variables. This reformulation will then serve as the basis for distributed decomposition.
Based on this reformulation, we then develop a distributed algorithm to solve the resulting reformulated optimization problem. 

We begin by considering a fixed and complete communication graph with edge set $\mathcal{E}(\mathbf{x}(t)) = \{(i,j)\in \mathcal{V}\times \mathcal{V}:i\neq j\}$, $\forall t\geq t_0$. This corresponds to the scenario in which the communication range $d_c$ is sufficiently large so that all agents can exchange information with one another. We subsequently extend the analysis to the case where the communication range $d_c$ is limited, leading to a time-varying communication topology.

For each pair of agents $(i,j)$ with $i\neq j$, we introduce mismatch variables $z_i^{i,j}, z_j^{i,j}\in \mathbb{R}$ into \eqref{eq:opt-inter} and define
\begin{align*}
    &{g}_{i}^{i,j}(u_i, z_i^{i,j},z_j^{i,j}; x_i, x_j)\triangleq h_{i,j}(x_i, x_j)(- L_{F_i}h_{\text{inter}}^{i,j}(x_i, x_j) \\
    &- L_{G_i}h_{\text{inter}}^{i,j}(x_i, x_j)\cdot u_i-\frac{1}{2}\gamma_1(h_{\text{inter}}^{i,j}(x_i, x_j)) + (z_i^{i,j}-z_j^{i,j})),
\end{align*}
where  
\begin{align*}
    h_{i,j}(x_i, x_j) \triangleq \left\{
    \begin{array}{ll}
    e^{-\frac{1}{d_c - \|p_i-p_j\|}}, & d_c - \|p_i-p_j\|>0\\
    0, & d_c - \|p_i-p_j\|\leq 0
    \end{array}\right. 
\end{align*}
Recasting \eqref{eq:opt-inter} with the mismatch variables, we obtain
\begin{align}
    \label{eq:recast-inter-agent-cbf}
    \left\{
    \begin{array}{ll}
     {g}_{i}^{i,j}(u_i, z_i^{i,j},z_j^{i,j}; x_i, x_j)\leq 0    &  \\
      {g}_{j}^{i,j}(u_j, z_j^{i,j},z_i^{i,j}; x_j, x_i)\leq 0   & 
    \end{array}\right.
\end{align}
Similarly, we introduce mismatch variables $z_i^c\in \mathbb{R}$ into $g^c(\mathbf{x}, \mathbf{u})$. Let $\mathbf{z}^c \triangleq \text{col}(\{z_i^c\}_{i=1}^{N})$, and $\mathbf{z}_i^c \triangleq \text{col}(\{z_j^c:j\in \{i\}\cup \mathcal{N}_i(\mathbf{x})\})$ denote the vector composed of the variables $z_j^c$ accessible to agent $i$. Define $\mathbf{x}_{\mathcal{N}_i} \triangleq \text{col}(\{x_j:j\in \{i\}\cup \mathcal{N}_i(\mathbf{x}) \})$ as the vector of states accessible to agent $i$. 
Given a sufficiently large $d_c$, $\mathbf{x}$ and $\mathbf{z}^c$ are accessible to each agent. 
We define
\begin{align*}
    &{g}_i^{c}(u_i, \mathbf{z}^c; \mathbf{x})  \triangleq -L_{F_i}h_c(\mathbf{x}) - L_{G_i}h_c(\mathbf{x})\cdot u_i \\
    &-\gamma_4\cdot \frac{1}{N}h_c(\mathbf{x})+\sum_{j\neq i}h_{i,j}(x_i, x_j)(z_i^c - z_j^c).   
\end{align*}
The constraint is represented as ${g}_i^{c}(u_i, \mathbf{z}^c; \mathbf{x})\leq 0$ for all $i\in\mathcal{V}$.
\begin{assumption}
\label{assumption:connectivity}
At each time $t$, each agent $i\in \mathcal{V}$ has instantaneous 
    access to the connectivity function $h_c(\mathbf{x})$ and its partial derivative $\frac{\partial h_c(\mathbf{x})}{\partial x_i}$ based on its own state and the state information of its neighbors.
\end{assumption}

In \cite{capelli2020connectivity}, a connectivity 
function $h_c(\mathbf{x})$ is constructed for 
single-integrator multi-agent systems. Following 
\cite{yang2010decentralized, malli2021robust}, both 
$h_c(\mathbf{x})$ and 
$\frac{\partial h_c(\mathbf{x})}{\partial \mathbf{x}}$ 
are computed using distributed approximation algorithms. We therefore assume that the distributed 
estimates of $h_c(\mathbf{x})$ and 
$\frac{\partial h_c(\mathbf{x})}{\partial \mathbf{x}}$ 
are sufficiently accurate, so that the conditions in 
Assumption~\ref{assumption:connectivity} are satisfied.

Next, we introduce mismatch variable $z_i^{clf}\in \mathbb{R}$ into \eqref{eq:clf}. Let   $\mathbf{z}^{clf}\triangleq \text{col}(\{z_i^{clf}\}_{i=1}^{N})$. We define 
\begin{align*}
    &{g}_i^{clf}(u_i, \mathbf{z}^{clf}; \mathbf{x}) \triangleq L_{F_i}(h_i(x_i)V_i(x_i)) \\
    &+ L_{G_i}(h_i(x_i)V_i(x_i)) \cdot u_i+ \gamma_3(h_i(x_i) V_i(x_i)) \\
    &+\sum_{j\neq i}h_{i,j}(x_i, x_j)h_i(x_i)h_j(x_j)(z_i^{clf}-z_j^{clf}),
\end{align*}
    and decompose \eqref{eq:clf} as ${g}_i^{clf}(u_i, \mathbf{z}^{clf}; \mathbf{x})  \leq 0, \forall i\in \mathcal{V}$.
We define the stacked mismatch variable 
$\mathbf{z}\triangleq \text{col}(\{z_i^{i,j}, z_j^{i,j}: (i,j)\in \overline{\mathcal{E}}\}, \mathbf{z}^c, \mathbf{z}^{clf})$.
The recast problem \eqref{eq:problem-formulation} can be formulated as
\begin{subequations}
\label{eq:recast-problem-formulation}
\begin{align}    
    \min_{\mathbf{u}, \mathbf{z}}&\quad \frac{1}{2}(\|\mathbf{u}\|^2 +\xi\cdot \|\mathbf{z}\|^2)\\
        \text{s.t.} 
        & \quad {g}_{i}^{i,j}(u_i, z_i^{i,j},z_j^{i,j}; x_i, x_j)\leq 0, \forall i\neq j \\
        &\quad -g_{\text{obs}}^{i,k}(u_i; x_i)\leq 0, \forall i\in \mathcal{V}, k\in [K] \label{eq:18-obs}\\
        &\quad {g}_i^c(u_i, \mathbf{z}^c; \mathbf{x})\leq 0, \forall i\in \mathcal{V} \\
     &\quad {g}_i^{clf}(u_i, \mathbf{z}^{clf}; \mathbf{x})  \leq 0, \forall i\in \mathcal{V}\\
        &\quad \|u_i\|_{\infty}\leq c, \forall i\in \mathcal{V}
\end{align}
\end{subequations}
where $\xi>0$.
For a given $\mathbf{x}$, the cost and constraint functions in \eqref{eq:recast-problem-formulation} are convex in $(\mathbf{u}, \mathbf{z})$, with the cost function being strictly convex.
Based on \cite[Proposition~4.1]{mestres2023distributed}, we have the following proposition.
\begin{proposition}[Problem Equivalence \cite{mestres2023distributed}]
\label{propisition-equivalence}
Problem \eqref{eq:problem-formulation} is equivalent to 
Problem \eqref{eq:recast-problem-formulation}.
Let $(\mathbf{u}^{\ast}, \mathbf{z}^{\ast})$ be the optimal solution of \eqref{eq:recast-problem-formulation}. Then $\mathbf{u}^{\ast}$ is the optimizer of \eqref{eq:problem-formulation} .
\end{proposition} 

For a given $\mathbf{z}$, let $\mathbf{z}^{\{i\}} \triangleq \text{col}(\{z_i^{i,j}, z_j^{i,j}\}_{j\in \mathcal{V}\setminus\{i\}}, \mathbf{z}^{c}, \mathbf{z}^{clf})$, and define $\overline{u}_i(\mathbf{x}, \mathbf{z}^{\{i\}})$ as the solution to the following optimization problem.
\begin{subequations}
    \label{eq:local-problem-formulation}
\begin{align}
    \min_{u_i}&\quad \frac{1}{2}\|u_i\|^2 \\
        \text{s.t.} 
        & \quad {g}_{i}^{i,j}(u_i; z_i^{i,j},z_j^{i,j}, x_i, x_j)\leq 0, \forall j\neq i \label{eq:19-b}\\
        &\quad -g_{\text{obs}}^{i,k}(u_i; x_i)\leq 0, \forall  k\in [K] \label{eq:19-c}\\
        &\quad  {g}_i^c(u_i; \mathbf{z}^c, \mathbf{x})\leq 0,  \label{eq:19-d}\\
     &\quad {g}_i^{clf}(u_i; \mathbf{z}^{clf}, \mathbf{x})  \leq 0, \label{eq:19-e}\\
        &\quad\begin{bmatrix}
            I^m\\-I^m
        \end{bmatrix}u_i\leq \begin{bmatrix}
            c \mathbf{1}^m\\c \mathbf{1}^m\label{eq:19-f}
        \end{bmatrix}
\end{align}
\end{subequations}
Define $\overline{\mathbf{u}}(\mathbf{x}, \mathbf{z}) \triangleq \text{col}(\{\overline{u}_i(\mathbf{x}, \mathbf{z}^{\{i\}}): i\in \mathcal{V}\})$. If $\mathbf{z} = \mathbf{z}^{\ast}$, then we have $\overline{\mathbf{u}}(\mathbf{x}, \mathbf{z}^{\ast}) = \mathbf{u}^{\ast}$. 
\begin{assumption}
\label{assumption:compute_u_bar}
At each time $t$, each agent $i\in \mathcal{V}$ has instantaneous access to the value of $\overline{u}_i$.
\end{assumption}

This assumption is also adopted in \cite{mestres2024distributed}. 
It is practically motivated since, in many cases of interest, \eqref{eq:local-problem-formulation} reduces to a quadratic program, which can be solved efficiently using the method in  \cite{stellato2020osqp}.

Stacking all constraint functions of agent $i$ in \eqref{eq:local-problem-formulation}, we write them compactly in matrix form as $g_i(u_i; \mathbf{z}^{\{i\}},\mathbf{x})\triangleq \Psi_i(\mathbf{x})\cdot u_i+{\Theta}_i(\mathbf{x})\cdot \mathbf{z}^{\{i\}}+\phi_i(\mathbf{x})\leq 0$.
The definitions of $\Psi_i(\mathbf{x})$, ${\Theta}_i(\mathbf{x})$, and $\phi_i(\mathbf{x})$ are included in Appendix \ref{appendix:function-construction}. 
By construction, we have $g_i(u_i; \mathbf{z}^{(i)},\mathbf{x})\in \mathbb{R}^{M_i}$ with $M_i = N+K+1+2m$, 
and $\Psi_i(\mathbf{x})\in \mathbb{R}^{M_i\times m}$, ${\Theta}_i(\mathbf{x})\in \mathbb{R}^{M_i\times (4N-2)}$, and $\phi_i(\mathbf{x})\in \mathbb{R}^{M_i}$. Moreover, the constraints in \eqref{eq:recast-problem-formulation} can be represented as
$g_i(u_i,  \mathbf{z}^{\{i\}};\mathbf{x})\leq 0, ~ \forall i\in \mathcal{V}.$

Let $y_i^{i,j}\geq 0$, $y_{obs}^{i,k}\geq 0$, $y_i^c\geq 0$, $y_i^{clf}\geq 0$ and $y_i^{u}\in \mathbb{R}_{+}^{2m}$ denote the dual variables associated with the constraints \eqref{eq:19-b} to \eqref{eq:19-f}, respectively. 
Let $\tau>0$ denote the time-scale parameter. The proposed distributed controller is given by
\begin{subequations}
\label{eq:two-scale-dynamics}
\begin{align}
\dot{x}_i &= F_i(x_i)+G_i(x_i)\cdot \overline{u}_i(\mathbf{x}, \mathbf{z}^{\{i\}})  \label{eq:slow-1}\\
\tau \dot{w}_i &= -w_i-\Psi_i(\mathbf{x})^Ty_i \label{eq:fast-1}\\
\tau \dot{z}_i^{i,j} &=-\xi {z}_i^{i,j}-h_{i,j}(x_i, x_j)(y_i^{i,j}-y_j^{i,j}) \\
\tau \dot{z}_i^c &= -\xi z_i^c - \sum_{j\neq i}h_{i,j}(x_i, x_j)(y_i^{c}-y_j^{c})\\
\tau \dot{z}_i^{clf} &= -\xi z_i^{clf} - \sum_{j\neq i}h_{i,j}(x_i, x_j)h_i(x_i)h_j(x_j)(y_i^{clf}-y_j^{clf})\\
\tau \dot{y}_i & = \big[ g_i(w_i, \mathbf{z}^{\{i\}};\mathbf{x})\big]_{y_i}^{+} \label{eq:fast-last}
\end{align}
\end{subequations}
for all $i\in \mathcal{V}$, $j\in \mathcal{N}_i(\mathbf{x})$, where $w_i\in \mathbb{R}^m$ and $y_i = \text{col}(\{y_i^{i,j}\}_{j\neq i}, \{y_{obs}^{i,k}\}_{k\in [K]}, y_i^c, y_i^{clf}, y_i^{u})$.

The entire closed-loop control policy can be described as follows. We define $\mathbf{w} \triangleq \text{col}(\{w_i\}_{i\in \mathcal{V}})$ and $\mathbf{y} \triangleq \text{col}(\{y_i\}_{i\in \mathcal{V}})$.
At each time instant, each agent $i$  evaluates the value of $(\mathbf{w}, \mathbf{z}, \mathbf{y})$ in the fast system from  \eqref{eq:fast-1} to \eqref{eq:fast-last} and computes the associated $\overline{u}_i$ by solving \eqref{eq:local-problem-formulation}. 
Then, agent $i$ applies the control input $\overline{u}_i$ to the system \eqref{eq:slow-1}. 
\begin{assumption}
    \label{assumption:feasible-mismatch}
    The optimization problem \eqref{eq:local-problem-formulation} is feasible for all $\mathbf{x}\in \mathcal{X}$ and $\mathbf{z}^{\{i\}}\in \mathbb{R}^{4N-2}$.
\end{assumption}

For a fixed $\mathbf{x}$, the existence and uniqueness of the equilibrium point for \eqref{eq:fast-1} to \eqref{eq:fast-last} follows from \cite[Theorem~5.1]{cherukuri2017role}. We denote this equilibrium point by $\text{col}(\mathbf{w}^{\ast}(\mathbf{x}), \mathbf{z}^{\ast}(\mathbf{x}), \mathbf{y}^{\ast}(\mathbf{x}))$.
\begin{assumption}
\label{assumption:locally-lipschitz}
    The function $\overline{u}_i(\mathbf{x}, \mathbf{z}^{\{i\}})$  is locally Lipschitz in $(\mathbf{x}, \mathbf{z}^{\{i\}})$ for all $i\in \mathcal{V}$. Furthermore, for given $\mathbf{x}$, $\text{col}(\mathbf{w}^{\ast}(\mathbf{x}), \mathbf{z}^{\ast}(\mathbf{x}), \mathbf{y}^{\ast}(\mathbf{x}))$ is locally Lipschitz in $\mathbf{x}$.
\end{assumption}

Conditions that guarantee the local Lipschitz continuity of the solution maps of parametric optimization problems, including those of the form \eqref{eq:local-problem-formulation} and \eqref{eq:recast-problem-formulation}, are established in \cite{mestres2025regularity}.

The two-timescale dynamical system \eqref{eq:two-scale-dynamics} is a special case of the differential inclusion \eqref{eq:diff-inclusion}. 
Let the fast-variable vector in \eqref{eq:two-scale-dynamics} be defined as
$\boldsymbol{\zeta} \triangleq \text{col}(\mathbf{w},\mathbf{z},\mathbf{y})$, 
and denote the equilibrium point of the fast variables by
$\boldsymbol{\zeta}^{\ast}(\mathbf{x}) 
\triangleq 
\text{col}\left(
\mathbf{w}^{\ast}(\mathbf{x}),
\mathbf{z}^{\ast}(\mathbf{x}),
\mathbf{y}^{\ast}(\mathbf{x})
\right).$
Under this notation, the set-valued map appearing in \eqref{eq:diff-inclusion} reduces to a singleton set for the system \eqref{eq:two-scale-dynamics}. 
Specifically, for each $(\mathbf{x},\boldsymbol{\zeta})$, the right-hand side of \eqref{eq:diff-inclusion} satisfies
\[
J(\mathbf{x},\boldsymbol{\zeta})
=
\big\{
\mathcal{F}(\mathbf{x},\boldsymbol{\zeta})
\big\},
\]
where $\mathcal{F}(\mathbf{x},\boldsymbol{\zeta})$ denotes the vector field defining the dynamics in \eqref{eq:two-scale-dynamics}. 
It follows that \eqref{eq:two-scale-dynamics} is an ordinary differential equation that corresponds to the unique realization of the differential inclusion \eqref{eq:diff-inclusion}.
Moreover, the set-valued map $\mathcal{H}(\cdot)$ in \eqref{eq:slow-system} is given by the singleton
$\mathcal{H}(\mathbf{x})
\triangleq
\left\{
\text{col}\left(
\mathbf{w}^{\ast}(\mathbf{x}),
\mathbf{z}^{\ast}(\mathbf{x}),
\mathbf{y}^{\ast}(\mathbf{x})
\right)
\right\}$.
The slow vector field $D(\cdot)$ appearing in \eqref{eq:slow-system} is defined as
$D(\mathbf{x})
\triangleq
\text{col}\left(
D_1(\mathbf{x}),\ldots,D_N(\mathbf{x})
\right)$,
where, for each agent $i$,
$D_i(\mathbf{x})
\triangleq
F_i(x_i)
+
G_i(x_i)\,
\overline{u}_i\!\left(\mathbf{x},(\mathbf{z}^{\{i\}})^{\ast}\right)$.

With these definitions, the reduced slow dynamics obtained from \eqref{eq:two-scale-dynamics} coincide with the system \eqref{eq:slow-system}, and the conditions required for applying Theorem~\ref{:theorem-converge} developed for \eqref{eq:diff-inclusion} are satisfied. For the system \eqref{eq:two-scale-dynamics}, it suffices to verify Assumptions HF, H1, H2, and H3 required by Theorem~\ref{:theorem-converge}. 
Assumptions HF, H2, and H3 are established in Appendix~\ref{appendix:lemmas} through Lemma \ref{lemma:hf} to Lemma \ref{lemma:h3}. Assumption H1 is verified in the following theorem.

We now state the main result of this paper. 
\begin{theorem}
\label{theorem:safe-stable}
    If Assumptions  \ref{assumption:feasible} to  \ref{assumption:locally-lipschitz} hold, then the controller $\overline{u}_i(\mathbf{x}, \mathbf{z}^{\{i\}})$ is safe and stable for all $i\in \mathcal{V}$, i.e., 
    \begin{enumerate}
        \item $\|p_i(t)-p_j(t)\|\geq d_0,     \forall i\neq j$ and $\|p_i(t)-s_k\|\geq r_k,  \forall i\in \mathcal{V}, k\in [K]$ for all $t\geq t_0$.
        \item The graph $\mathcal{G}(t)$ is connected for all $t\geq t_0$.
        \item For any $\mathbf{x}(t_0)=\mathbf{x}_0\in \mathcal{X}$,
there exists $r>0$ with the following property.
Suppose
$$\text{col}(\mathbf{w}_0,\mathbf{z}_0, \mathbf{y}_0)
\in \{\text{col}(\mathbf{w}^{\ast}(\mathbf{x}_0),
\mathbf{z}^{\ast}(\mathbf{x}_0),
\mathbf{y}^{\ast}(\mathbf{x}_0))\} + B_r.$$
Then, for every $\eta>0$, there exist $T_{\eta}>t_0$
and, for each $T\geq T_{\eta}$, a constant $\tau_0>0$ can be found
such that
$$\sup_{t\geq T} d_A(\mathbf{x}(t))\leq \eta,
\quad \forall \tau \in (0,\tau_0).$$
    \end{enumerate}
\end{theorem}

\begin{proof}
    We first prove the safety of the controller. At each time instant, for a fixed $\mathbf{z}$, the solution to \eqref{eq:local-problem-formulation} is denoted as $\overline{\mathbf{u}}(\mathbf{x}, \mathbf{z})$. Then, we have $\text{col}(\overline{\mathbf{u}}( \mathbf{x},\mathbf{z}), \mathbf{z})$ is in the feasible set of \eqref{eq:recast-problem-formulation}, since the constraints in \eqref{eq:recast-problem-formulation} are satisfied by construction. Furthermore, we have that 
    \begin{align*}
        -{g}_{\text{inter}}^{i,j}(u_i, u_j; x_i, x_j)& = {g}_{i}^{i,j}(u_i, z_i^{i,j},z_j^{i,j}; x_i, x_j)\\
        &+{g}_{j}^{j,i}(u_j, z_j^{i,j},z_i^{i,j}; x_j, x_i) \leq 0, \forall i \neq j,
    \end{align*}
    $-g_{\text{obs}}^{i,k}(u_i;x_i)\leq 0$, $\forall k\in [K]$, for all $i\in \mathcal{V}$, and $-g^c(\mathbf{x}, \mathbf{u}) =\sum_{i=1}^{N}{g}_i^{c}(u_i, \mathbf{z}^c; \mathbf{x})\leq 0$, implying \eqref{eq:opt-inter}, \eqref{eq:16-obs}, and \eqref{eq:opt-connectivity} are satisfied for all $t\geq t_0$. 
    By invoking Theorem~\ref{theorem:cbf}, the super-level sets
sets $\{(x_i, x_j): h_{\text{inter}}^{i,j}(x_i, x_j) \geq 0\}$ for all $i\neq j$, $\{x_i: h_{\text{obs}}^{i,k}(x_i) \geq 0\}$ for all $i \in \mathcal{V}$, $k \in [K]$, and $\{\mathbf{x}: h_c(\mathbf{x}) \geq 0\}$ are all forward invariant. By construction of $h_{\text{inter}}^{i,j}(x_i, x_j)$ and
$h_{\text{obs}}^{i,k}(x_i)$, this implies
$\|p_i-p_j\|\geq d_0$ for all $i\neq j$ and
$\|p_i-s_k\|\geq r_k$ for all $i\in \mathcal{V}$,
$k\in [K]$, for all $t\geq t_0$. Since $\{\mathbf{x}: h_c(\mathbf{x}) \geq 0\} \subseteq \{\mathbf{x}: \lambda_2(\mathbf{x}) > 0\}$, it follows that $\lambda_2(\mathbf{x}(t)) > 0$ for all $t \geq t_0$, preserving connectivity of the communication graph.

    Next, we prove the stability of the controller. 
    Since the closed-loop dynamics
\eqref{eq:two-scale-dynamics} are non-smooth, we appeal to
Theorem~\ref{:theorem-converge}.
For any fixed $\mathbf{x}$, system \eqref{eq:fast-1} to \eqref{eq:fast-last} has a unique equilibrium point $\text{col}(\mathbf{w}^{\ast}(\mathbf{x}), \mathbf{z}^{\ast}(\mathbf{x}), \mathbf{y}^{\ast}(\mathbf{x}))$ \cite[Theorem~5.1]{cherukuri2017role}. It follows that the equilibrium set is a singleton, hence trivially compact and convex. Moreover, by Assumption \ref{assumption:locally-lipschitz}, $\text{col}(\mathbf{w}^{\ast}(\mathbf{x}), \mathbf{z}^{\ast}(\mathbf{x}), \mathbf{y}^{\ast}(\mathbf{x}))$ is locally Lipschitz in $\mathbf{x}$. Thus, Assumption H1 is satisfied.
By Lemmas~\ref{lemma:hf} through~\ref{lemma:h3}, all remaining assumptions
of Theorem~\ref{:theorem-converge} are satisfied for the
closed-loop system. Applying Theorem~\ref{:theorem-converge} then
yields that, for sufficiently small $\tau$, the state
$\mathbf{x}(t)$ converges to a neighborhood of $A^{\epsilon}$. Since $A\subseteq A^{\epsilon}$, it follows that
$\sup_{t\geq T} d_A(\mathbf{x}(t))\leq
\sup_{t\geq T} d_{A^{\epsilon}}(\mathbf{x}(t))\leq\eta$,
completing the proof.
\end{proof}

Theorem~\ref{theorem:safe-stable} demonstrates that the controller $\overline{u}_i(\mathbf{z}, \mathbf{x})$ guarantees both safety and convergence of the agents to the target region while preserving connectivity. However, this controller requires access to global information $(\mathbf{x}, \mathbf{z}^c, \mathbf{z}^{clf})$ for its computation. When the communication range $d_c$ is limited, or when hardware constraints prevent each agent from storing or processing global information in large-scale systems, a scalable implementation of the controller becomes necessary. 

For each agent $i$, computing $\overline{u}_i$ by solving  \eqref{eq:local-problem-formulation} requires both local computation and information exchanged with neighboring agents. 
At the initialization stage, agent $i$ initializes and stores the following fast variables $\{z_i^{i,j}, z_j^{i,j}: j\neq i\}$, $\{y_i^{i,j}, y_j^{i,j}: j\neq i\}$, $z_i^c$, $y_i^c$, $z_i^{clf}$, and $y_i^{clf}$. In particular, all variables $z_i^{i,j}$ and $y_i^{i,j}$ are initialized identically across agents.

For the communication constraint \eqref{eq:19-d}, define $\mathbf{z}_{\mathcal{N}_i}^{c} \triangleq \text{col}(\{z_i^c: \{i\}\cup \mathcal{N}_i(\mathbf{x})\})$, we have
\begin{align*}
    &{g}_i^c(u_i; \mathbf{z}^c, \mathbf{x}) = -L_{F_i}h_c(\mathbf{x}) - L_{G_i}h_c(\mathbf{x})\cdot u_i \\
    &-\gamma_4\cdot \frac{1}{N}h_c(\mathbf{x})+\sum_{j\in \mathcal{N}_i(\mathbf{x})}h_{i,j}(x_i, x_j)(z_i^c - z_j^c),   
\end{align*}
This expression depends only on $(\mathbf{x}_{\mathcal{N}_i}, \mathbf{z}_{\mathcal{N}_i}^{c})$.
The update dynamics of $z_i^c$ and $y_i^c$ are given by
\begin{align*}
    \tau \dot{z}_i^c 
    & = -\xi z_i^c - \sum_{j\in \mathcal{N}_i(\mathbf{x})}h_{i,j}(x_i, x_j)(y_i^{c}-y_j^{c})\\
    \tau \dot{y}_i^c& = [{g}_i^c(u_i; \mathbf{z}^c, \mathbf{x})]_{y_i^c}^{+},
\end{align*}
Hence, given
$\mathbf{x}_{\mathcal{N}_i}$, $(z_i^c,y_i^c)$ can be updated locally using information received from neighboring agents, namely $(z_j^c, y_j^c)$ for all $j\in \mathcal{N}_i(\mathbf{x})$.
Similarly, define $\mathbf{z}_{\mathcal{N}_i}^{clf} \triangleq \text{col}(\{z_i^{clf}: \{i\}\cup \mathcal{N}_i(\mathbf{x})\})$. The constraint \eqref{eq:19-e} can be represented as
\begin{align*}
    &{g}_i^{clf}(u_i, \mathbf{z}^{clf}; \mathbf{x}) = -L_{F_i}(h_i(x_i)V_i(x_i)) \\
    &- L_{G_i}(h_i(x_i)V_i(x_i)) \cdot u_i- \gamma_3(h(x_i) V_i(x_i)) \\
    &+\sum_{j\in \mathcal{N}_i(\mathbf{x})}h_{i,j}(x_i, x_j)(z_i^{clf}-z_j^{clf}).
\end{align*}
Hence, given  $\mathbf{x}_{\mathcal{N}_i}$, $(z_i^{clf},y_i^{clf})$ can be updated locally with information exchanged from neighbors $(z_j^{clf}, y_j^{clf})$ for all $j\in \mathcal{N}_i(\mathbf{x})$.
Finally, the constraint \eqref{eq:19-b} satisfies $g_i^{i,j}(u_i; z_i^{i,j}, z_j^{i,j}, x_i, x_j) = 0$, for all $j\notin \mathcal{N}_i(\mathbf{x})$, implying \eqref{eq:19-b} is equivalent to 
\begin{align*}
    g_i^{i,j}(u_i; z_i^{i,j}, z_j^{i,j}, x_i, x_j) \leq  0, \forall j\in \mathcal{N}_i(\mathbf{x}).
\end{align*}
When $j\in \mathcal{N}_i(\mathbf{x})$, agent $i$ receive $(z_j^{i,j}, y_j^{i,j})$ from agent $j$ and update $(z_i^{i,j}, y_i^{i,j})$ locally. When $j\notin \mathcal{N}_i(\mathbf{x})$, the dynamics reduce to $\tau \dot{z}_i^{i,j} = -\xi z_i^{i,j},\;
\tau \dot{z}_j^{i,j} = -\xi z_j^{i,j},\;
\tau \dot{y}_i^{i,j} = \tau \dot{y}_j^{i,j} = 0.$
The value of $(z_j^{i,j}, y_j^{i,j})$ can be evaluated by agent $i$, yielding the same values as those computed by agent $j$.
Therefore, all quantities required to compute $\overline{u}_i$ can be obtained using only local information and communication with neighboring agents. Consequently, the proposed control law can be implemented in a fully distributed manner.

\section{Simulations}
\label{section:simulation}
In this section, we validate the effectiveness of the proposed approach through a case study. All simulations are conducted on a laptop with an Apple M3 processor and 36 GB of RAM.
\subsection{Setup}
We consider a multi-agent system consisting of 5 agents. Each agent follows single integrator dynamics $\dot{x}_i(t) = u_i(t)$, where $x_i = [x_{i,1}, x_{i,2}]^T\in \mathcal{X}\subseteq\mathbb{R}^2$ and $u_i\in \mathcal{U}\subseteq\mathbb{R}^2$ denote the position and control input of agent $i$, respectively. The target region for each agent is defined as $\Omega_i = \{x_i\in \mathbb{R}^2: x_i^T\Sigma_ix_i\leq 1\}$, where $\Sigma_i = \begin{bmatrix}
    2 &0\\
    0 & 0.25
\end{bmatrix}$ for all $i\in \mathcal{V}$. The agents are initially placed at $[-2.1, -0.75]^T$, $ [-2.1, 0.75]^T$, $[-2.0, 0.0]^T$, $[-1.9, -0.5]^T$, and $[-1.9, 0.5]^T$. This scenario includes seven obstacles, with centers at $[-1.5, 0.6]^T$, $[-1.1, 0.72]^T$, $[-0.65, 1.0]^T$, $[-1.25, 0.05]^T$, $[-1.5, -0.6]^T$, $[-1.1, -0.72]^T$, and $[-0.65, -1.0]^T$. The corresponding radii for each obstacles are $0.16$, $0.22$, $0.2$, $0.1$, $0.16$, $0.22$, and $0.2$, respectively.

We construct CBFs to enforce safety and connectivity constraints as follows. For inter-agent collision avoidance, we define $g_{\text{inter}}^{i,j}(u_i, u_j; x_i, x_j)=2(x_i-x_j)^T(u_i-u_j)+\gamma_1(\|x_i-x_j\|^2-d_0^2)$, where $\gamma_1 = 2$. For obstacle collision avoidance, we define $g_{\text{obs}}^{i,k}(u_i; x_i)= 2(x_i-s_k)^Tu_i+\gamma_2(\|x_i-s_k\|^2-r_k^2)$, where $\gamma_2 = 2$. To enforce connectivity, we define a weighted communication graph with adjacency weights
\begin{align*}
    &a_{i,j} \\
    &\triangleq \left\{
    \begin{array}{ll}
     e^{((d_c-\epsilon_c)^2 - \|x_i-x_j\|^2)^2/\sigma} -1,   &  \|x_i-x_j\|\leq d_c-\epsilon_c\\
     0,    & \text{otherwise}
    \end{array}\right.
\end{align*}
where $\sigma>0$ is chosen to ensure $a_{i,j}\leq 1$. In this simulation, we set $d_c = 0.9$, $\epsilon_c=0.1$, and $\sigma = \frac{(d_c-\epsilon_c)^4}{\ln 2}$.
The corresponding Laplacian matrix $L_{comm}$ is defined as 
\begin{align*}
    [L_{comm}]_{i,j}\triangleq 
    \left\{
    \begin{array}{ll}
     \sum_{j=1}^N a_{i,j},    & i=j \\
     -a_{i,j},    & i\neq j
    \end{array}\right.
\end{align*}
Let $\lambda_2(\mathbf{x})$ denote the second smallest eigenvalue of $L_{comm}$. The connectivity CBF is defined as $h_c(\mathbf{x}) = \lambda_2(\mathbf{x})-\chi$, where $\chi = 0.1$. The corresponding CBF constraint is $g^c(\mathbf{x}, \mathbf{u}) = \frac{\partial \lambda_2(\mathbf{x})}{\partial \mathbf{x}}\mathbf{u}+\lambda_2(\mathbf{x})-\chi$, where $\frac{\partial \lambda_2(\mathbf{x})}{\partial x_i} = \sum_{j\neq i}\frac{\partial a_{i,j}}{\partial x_i}(v_2^i-v_j^i)^2$. We use $v_2^i$ and $v_2^j$ to denote the $i$-th and $j$-th component of eigenvector associated to $\lambda_2(\mathbf{x})$, respectively.
In order to synthesize $g_{clf}(\mathbf{u}; \mathbf{x})$, we select $V_i(x_i) = x_i^T\Sigma x_i -1 +\epsilon$, where $\epsilon = 0.5$. The corresponding CLF-like function is defined as $g_{clf}(\mathbf{u}; \mathbf{x}) = \sum_{i=1}^{N}\left(\frac{\partial (h_i(x_i)V_i(x_i))}{\partial x_i}u_i+\gamma_3h_i(x_i)V_i(x_i)\right)$, where $\gamma_3 = 0.01$.
In the distributed algorithm, we set $\tau = 0.002$ and $\xi = 0.5$ for the fast dynamical system.
\subsection{Results} 
As shown in Fig.~\ref{fig:trajectories}, all agents successfully converge to the target region while satisfying safety and communication connectivity constraints. 
To illustrate the evolution of the communication topology, we select four representative time instants at $t = 0, 50, 75, 250$, and visualize the corresponding communication edges. The communication graph is state-dependent and time-varying. From $t=0$ to $t=50$, edge $(3,4)$ is established as Agents $3$ and $4$ move within communication range. From $t=50$ to $t=75$, edge $(1,2)$ is temporarily lost as the inter-agent distance exceeds the communication threshold $d_c$. From $t=75$ to $t=250$, edge $(1,2)$ is re-established. These results demonstrate that the proposed method can accommodate time-varying communication graphs, where edges may be created or removed dynamically, while still ensuring convergence to the target region. To further analyze the connectivity of the communication graph during the entire process, we construct a binary communication graph based on the inter-agent distances. Specifically, for each pair of agents $(i,j)$, the adjacency weight is defined as
\begin{align*}
    \overline{a}_{ij} =
\begin{cases}
1, & \|x_i - x_j\| < d_c, \\
0, & \text{otherwise}.
\end{cases}
\end{align*}
The corresponding Laplacian matrix is given by 
\begin{align*}
    [\overline{L}]_{i,j}\triangleq 
    \left\{
    \begin{array}{ll}
     \sum_{j=1}^N \overline{a}_{i,j},    & i=j \\
     -\overline{a}_{i,j},    & i\neq j
    \end{array}\right.
\end{align*} We then compute the second smallest eigenvalue $\overline{\lambda}_2$ of $\overline{L}$, which characterizes the connectivity of the communication graph. As shown in Fig.~\ref{fig:connectivity}, the value of $\overline{\lambda}_2$ remains strictly positive throughout the entire process, indicating that the communication graph remains connected at all times. 
\begin{figure}
    \centering
    \includegraphics[width = 0.45\textwidth]{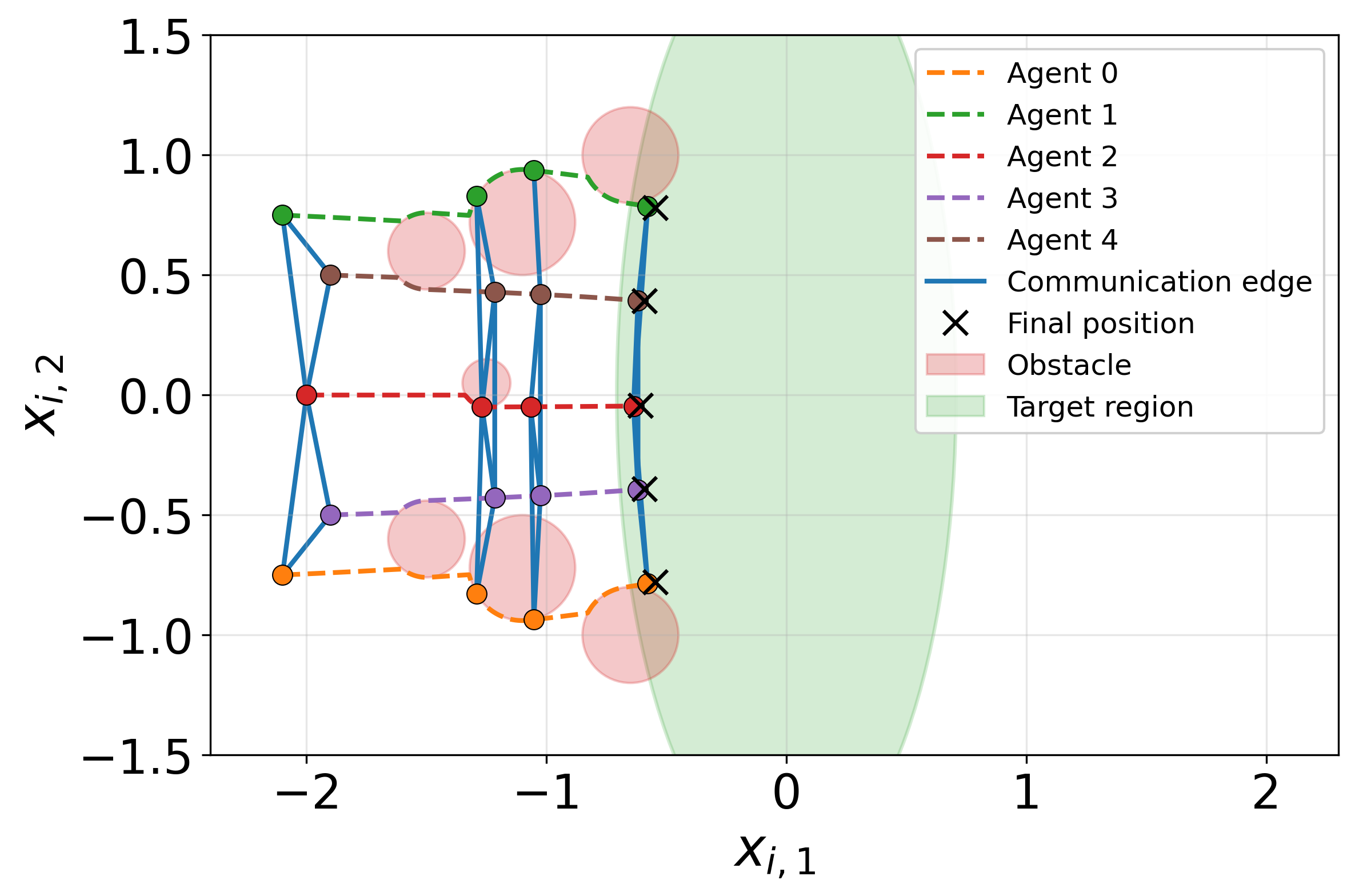}
    \caption{Trajectories of five agents and communication graphs at four selected time instants $t=0, 50, 75, 250$. Dashed curves denote agent trajectories, filled circles indicate agent positions at the selected instants, and crosses mark the final positions. Blue solid lines represent active communication edges. The green shaded region denotes the target region $\Omega_i$, and the light red circles represent obstacles. }       
    \label{fig:trajectories}
\end{figure}
\begin{figure}
    \centering
    \includegraphics[width = 0.45\textwidth]{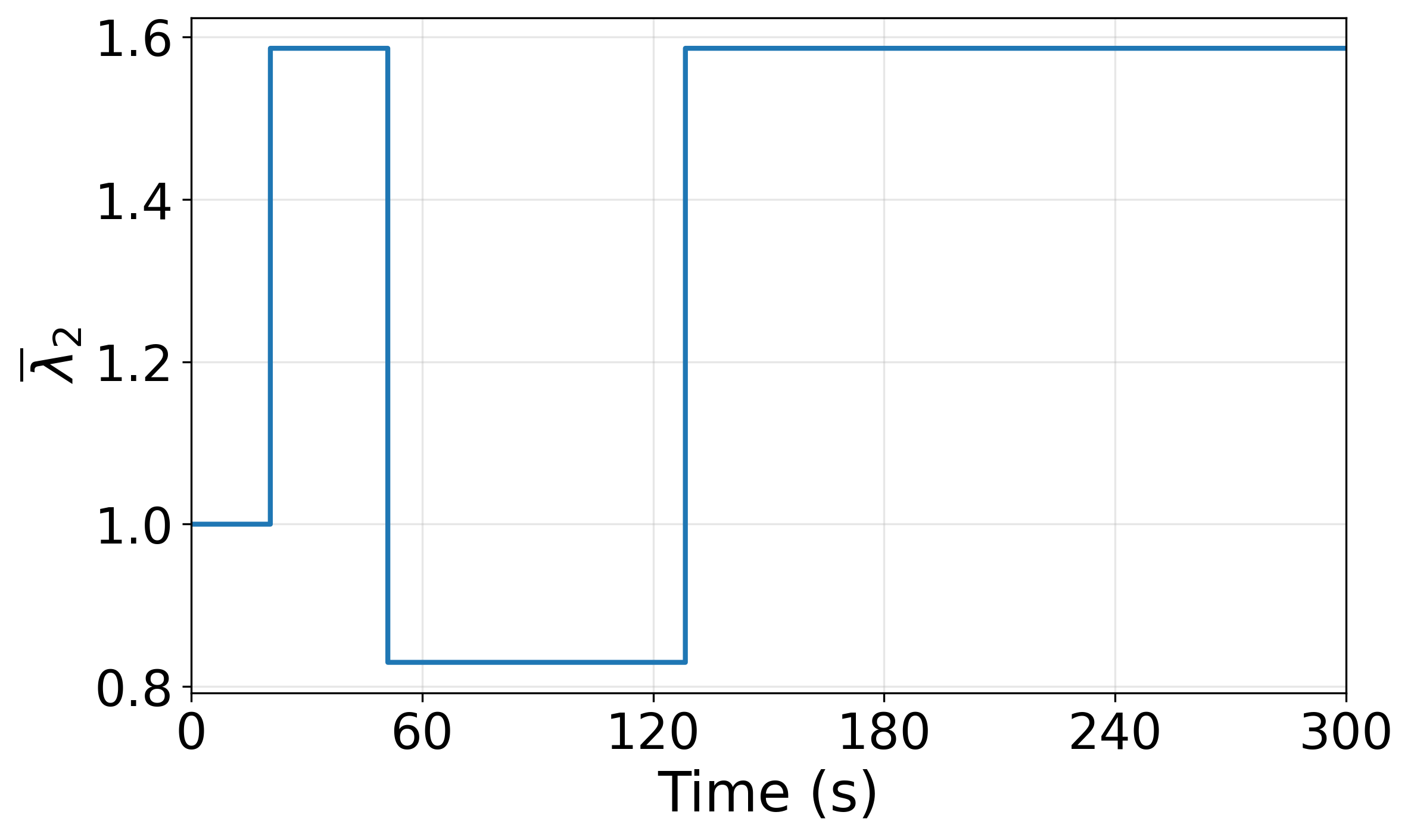}
    \caption{Evolution of the algebraic connectivity $\overline{\lambda}_2$ of the binary communication graph over time. The value of $\overline{\lambda}_2$ remains strictly positive throughout the process, indicating that the connectivity of the communication graph is maintained.}
    \label{fig:connectivity}
\end{figure}
\section{Conclusion}
\label{section:conclusion}

This paper presented a distributed optimization-based control framework for steering multi-agent systems to a target region while guaranteeing collision avoidance and communication connectivity under a state-dependent, time-varying communication topology. 
The framework decoupled globally coupled constraints through auxiliary mismatch variables evolved via two-time-scale dynamics, and a truncation function ensured that these variables remain well-defined as communication links appear or disappear. 
Using singular perturbation analysis, we established that the resulting distributed controller guarantees safety, connectivity preservation, and convergence to the target region. 
The framework was validated through numerical simulations on a multi-agent reach-avoid scenario with time-varying communication topologies. Simulation results confirm that all agents converge to the target region while safety and connectivity constraints are satisfied throughout the process.

\IEEEpeerreviewmaketitle

\bibliographystyle{IEEEtran}
\bibliography{MyBib}
\appendix

This appendix provides the technical details omitted from the main text. In particular, we present the construction of $g_i(u_i; \mathbf{z}^{(i)},\mathbf{x})$ in Appendix~\ref{appendix:function-construction}. Appendix~\ref{appendix:theorem3} contains the proof of Theorem \ref{theorem:clf-convergence}. Appendix C verifies Assumptions HF to H3 in Theorem \ref{:theorem-converge} through a sequence of auxiliary lemmas.

\subsection{Construction of $g_i(u_i; \mathbf{z}^{\{i\}},\mathbf{x})$}
\label{appendix:function-construction}
For brevity, we suppress the arguments and write $h_i$ for $h_i(x_i)$, $h_{i,j}$ for $h_{i,j}(x_i, x_j)$, $h_{\text{inter}}^{i,j}$ for $h_{\text{inter}}^{i,j}(x_i, x_j)$, and $h_{\text{obs}}^{i,k}$ for $h_{\text{obs}}^{i,k}(x_i)$. 
The constraints in \eqref{eq:local-problem-formulation} is defined as $g_i(u_i; \mathbf{z}^{\{i\}},\mathbf{x})=\Psi_i(\mathbf{x})\cdot u_i+{\Theta}_i(\mathbf{x})\cdot \mathbf{z}^{\{i\}}+\phi_i(\mathbf{x})$, 
where
\begin{align*}
    \Psi_i(\mathbf{x})
    \triangleq \bigg[ \text{col} &\bigg( 
        \{- h_{i,j}\cdot L_{G_i}h_{\text{inter}}^{i,j}\}_{ j\neq i}, \\
        &
         \{-L_{G_i}h_{\text{obs}}^{i,k}\}_{k\in[K]}, 
          - L_{G_i}h_c(\mathbf{x}),\\
          &
        L_{G_i}(h_iV_i(x_i))\bigg)^T,  
         I^m, 
         -I^m\bigg]^T,
    \end{align*}
\begin{align*}
    &\phi_i(\mathbf{x}) \triangleq \bigg[\text{col}\bigg( 
        \{- h_{i,j}\cdot \big(L_{F_i}h_{\text{inter}}^{i,j}+\frac{1}{2}\gamma_1(h_{\text{inter}}^{i,j})\big)\}_{j\neq i},\\
        &\{ - L_{F_i} h_{\text{obs}}^{i,k}- \gamma_2  (h_{\text{obs}}^{i,k})\}_{k\in [K]}, -L_{F_i}h_c(\mathbf{x})-\gamma_4\cdot \frac{1}{N}h_c(\mathbf{x}), \\
        &
         L_{F_i}(h_iV_i(x_i)) + \gamma_3(h_i V_i(x_i))\bigg)^T, 
         c \mathbf{1}^m, c \mathbf{1}^m\bigg]^T
    ,
    \end{align*}
\begin{align*}
    \Theta_i(\mathbf{x}) &= 
    \begin{bmatrix}
        D_h\cdot(I^{N-1}\otimes[1\ -1]) & 0^{(N-1)\times N} & 0^{(N-1)\times N} \\
        0^{K\times 2(N-1)} & 0^{K\times N} & 0^{K\times N} \\
        0^{1\times (2N-1)} & e_i^T\hat{L}_h(\mathbf{x}) & 0^{1\times N} \\
        0^{1\times (2N-1)} & 0^{1\times N} & e_i^T\tilde{L}_h(\mathbf{x})\\
        0^{2m\times 2(N-1)} & 0^{2m\times N} & 0^{2m\times N} 
    \end{bmatrix}
\end{align*}
with $D_h\triangleq \operatorname{diag}(\{h_{i,j}(x_i, x_j)\}_{j\neq i})$,
$e_i \in \mathbb{R}^N$ denotes the $i$-th canonical basis vector,
i.e., 
$e_i \triangleq 
[0,\dots,0,1,0,\dots,0]^T$, whose $i$-th entry equals one and all other entries are zero, 
\begin{align*}
    [\hat{L}_h]_{ij} =\left\{
    \begin{array}{ll}
     \sum_{j\neq i} h_{i,j}, & i=j,\\
- h_{i,j}, & i\neq j,
    \end{array}\right.
\end{align*}
\begin{align*}
    [\tilde{L}_h]_{ij} =\left\{
    \begin{array}{ll}
     \sum_{j\neq i} h_{i,j} h_i h_j, & i=j,\\
- h_{i,j} h_i h_j, & i\neq j.
    \end{array}\right.
\end{align*}
\subsection{Proof of Theorem 
\ref{theorem:clf-convergence}}
\label{appendix:theorem3}
\textit{Proof of Theorem \ref{theorem:clf-convergence}}: 
By Assumption~\ref{assumption:feasible}, a solution $\mathbf{u}(t)$ 
to \eqref{eq:problem-formulation} exists for all $t \geq t_0$. The controller $\mathbf{u}(t)$, as the solution to \eqref{eq:problem-formulation}, satisfies $-g_{\text{inter}}^{i,j}(u_i, u_j; x_i, x_j) \leq 0$ for all $(i,j) \in \mathcal{E}(\mathbf{x})$, $-g_{\text{obs}}^{i,k}(u_i; x_i) \leq 0$ for all $i \in \mathcal{V}$, $k \in [K]$, and $-g^c(\mathbf{u}; \mathbf{x}) \leq 0$. By Theorem~\ref{theorem:cbf}, the sets $\{(x_i, x_j): h_{\text{inter}}^{i,j}(x_i, x_j) \geq 0\}$ for all $(i,j) \in \mathcal{E}(\mathbf{x})$, $\{x_i: h_{\text{obs}}^{i,k}(x_i) \geq 0\}$ for all $i \in \mathcal{V}$, $k \in [K]$, and $\{\mathbf{x}: h_c(\mathbf{x}) \geq 0\}$ are all forward invariant, ensuring inter-agent safety and obstacle avoidance for all $t \geq t_0$. Moreover, since $\{\mathbf{x}: h_c(\mathbf{x}) \geq 0\} \subseteq \{\mathbf{x}: \lambda_2(\mathbf{x}) > 0\}$, it follows that $\lambda_2(\mathbf{x}(t)) > 0$ for all $t \geq t_0$, preserving connectivity of the communication graph.

In what follows, we prove the stability of the controller. Since $V(\mathbf{x})\geq 0$ and $\dot{V}(\mathbf{x})+\gamma_3 V(\mathbf{x})\leq 0$, the sublevel set $\overline{D} = \{\mathbf{x}\in \mathbb{R}^{nN}: V(\mathbf{x})\leq V(\mathbf{x}(t_0))\}$ is compact and positively invariant.
    We have $  \dot{V}(\mathbf{x})=0$ if and only if $V(\mathbf{x}) = 0$, i.e., $\mathbf{x}\in A^{\epsilon}$. Thus, set $A^{\epsilon}$ is the largest invariant set contained in $\{\mathbf{x}\in \overline{D}: \dot{V}(\mathbf{x})=0\}$. By LaSalle's invariance principle, we have $\lim_{t\rightarrow\infty}d_{A^{\epsilon}}(\mathbf{x}(t)) = 0$. Since $A^{\epsilon}\subset A$, there exists $T>0$ such that $\mathbf{x}(t)\in A$, $\forall t\geq T$, implying we have $\lim_{t\rightarrow\infty}d_{A}(\mathbf{x}(t)) = 0$. 

    Next, we will prove that the set $A^{\epsilon}$ is UGAS. We first prove that there exists $\psi(\cdot)$ as defined in Definition \ref{def:ugas}, such that for any $R>0$, any $\mathbf{x}(t_0)\in A^{\epsilon}+\overline{B}_R$, $d_{A^{\epsilon}}(\mathbf{x}(t))\leq \psi(R), \forall t\geq t_0$. We have that $ d_{\Omega_1^{\epsilon}}^2(x_1(t_0)) + \cdots + d_{\Omega_N^{\epsilon}}^2(x_N(t_0)) =d_{A^{\epsilon}}^2(\mathbf{x}(t_0)) \leq R^2$.
Based on the fact that $0\leq h_i(x_i)\leq 1$, we have 
    $V(\mathbf{x}(t))\leq \sum_{i=1}^{N}V_i(x_i(t))\leq \sum_{i=1}^{N}\beta_2 \cdot d_{\Omega_i^{\epsilon}}(x_i(t))$.    
    Therefore, we have  $V(\mathbf{x}(t))\leq  \sum_{i=1}^{N}\beta_2\cdot  d_{\Omega_i^{\epsilon}}(x_i(t_0))\leq \sum_{i=1}^{N}\beta_2\cdot \frac{R}{\sqrt{N}} = \beta_2\sqrt{N} R$.  
Define $\theta(d_{\Omega_i^{\epsilon}}(x_i(t))) = e^{-\frac{1}{\beta_1\cdot d_{\Omega_i^{\epsilon}}(x_i(t))}}\cdot \left(\beta_1\cdot d_{\Omega_i^{\epsilon}}(x_i(t))\right)$ with $d_{\Omega_i^{\epsilon}}(x_i(t))>0$. Based on the construction, $\theta(d_{\Omega_i}(x_i(t)))$ is a continuous function and strictly increasing in $d_{\Omega_i}(x_i(t))$. Therefore, we have $\theta(d_{\Omega_i^{\epsilon}}(x_i(t)))\leq e^{-\frac{1}{V_i(x_i(t))}}\cdot V_i(x_i(t))$.
Moreover, we have that $\theta(\cdot)$ is invertible and $\theta^{-1}(\cdot)$ is also a continuous and strictly increasing function. If $d_{A^{\epsilon}}(\mathbf{x}(t))>0$, based on the fact $\frac{d_{A^{\epsilon}}(\mathbf{x}(t))}{\sqrt{N}}\leq \max_{i\in \mathcal{V}}d_{\Omega_i^{\epsilon}}(x_i)$, we have
\begin{align*}
 &\theta(\frac{d_{A^{\epsilon}}(\mathbf{x})}{\sqrt{N}})\leq \theta(\max_{i\in \mathcal{V}}d_{\Omega_i^{\epsilon}}(x_i))\\
 &\quad\leq \sum_{i\in \{j\in \mathcal{V}:x_j\notin\Omega_i^{\epsilon}\}}\theta(d_{\Omega_j^{\epsilon}}(x_i))\leq V(\mathbf{x})\leq \beta_2\sqrt{N}R.
\end{align*}
This implies $\frac{d_{A^{\epsilon}}(\mathbf{x}(t))}{\sqrt{N}} \leq \theta^{-1}\left(\beta_2\sqrt{N}R\right)$.
Define $\psi(R) = \sqrt{N}\cdot \theta^{-1}\left(\beta_2\sqrt{N}R\right)$, we have $\psi(R)$ is a strictly increasing function in $R$ and $\lim_{R\rightarrow 0}\psi(R)=0$, completing the proof of the second part.

Finally, we prove the third part by showing for any $R>0$, any $\mathbf{x}(t_0)\in A^{\epsilon}+\overline{B}_R$, and any $\nu>0$, there exists $\mathcal{T}: \mathbb{R}_{+}\times 
    \mathbb{R}_{+}\rightarrow\mathbb{R}_{+}$ such that $d_{A^{\epsilon}}(x(t))\leq \nu, \forall t\geq \mathcal{T}(R, \nu)$. 
To characterize the convergence rate of agents toward the target region, we relate the Lyapunov function $V(\mathbf{x})$ to $d_{A^{\epsilon}}(\mathbf{x})$. In particular, we identify $c_1, c_2>0$ satisfying 
$$d_{A^{\epsilon}}(\mathbf{x})\leq R \implies V(\mathbf{x})\leq c_2,~V(\mathbf{x})\leq c_1\implies d_{A^{\epsilon}}(\mathbf{x})\leq \nu.$$
Hence, an upper bound on the time required for $V(\mathbf{x}(t))$ to decrease from $V(\mathbf{x}(t_0))$ to $c_1$ directly yields an upper bound on the time for agents to reach the state with $d_{A^{\epsilon}}(\mathbf{x})\leq \nu$. We first construct a constant $c_1$. For all $d_{\Omega_i^{\epsilon}}(x_i)>0$, we have $e^{-\frac{1}{\beta_1\cdot d_{\Omega_i^{\epsilon}}(x_i))}}\cdot \left(\beta_1\cdot d_{\Omega_i^{\epsilon}}(x_i)\right) \leq e^{-\frac{1}{V_i(x_i(t))}}\cdot V_i(x_i)$. This follows that  
\begin{align*}
    V(\mathbf{x}) &= \sum_{i=1}^{N}h_i(x_i)V_i(x_i) =\sum_{i\in \{j\in \mathcal{V}: x_j\notin \Omega_j^{\epsilon}\}} e^{-\frac{1}{V_i(x_i(t))}}\cdot V_i(x_i)\\
    &\geq \sum_{i\in \{j\in \mathcal{V}: x_j\notin \Omega_j^{\epsilon}\}}e^{-\frac{1}{\beta_1\cdot d_{\Omega_i^{\epsilon}}(x_i))}}\cdot \left(\beta_1\cdot d_{\Omega_i^{\epsilon}}(x_i)\right).
\end{align*}
If $V(\mathbf{x})\leq e^{-\frac{1}{\beta_1\cdot \nu/\sqrt{N}}}\cdot \left(\beta_1 \nu/\sqrt{N}\right)$, we have $e^{-\frac{1}{\beta_1\cdot d_{\Omega_i^{\epsilon}}(x_i))}}\cdot \left(\beta_1\cdot d_{\Omega_i^{\epsilon}}(x_i)\right) \leq e^{-\frac{1}{\beta_1\cdot \nu/\sqrt{N}}}\cdot \left(\beta_1 \nu/\sqrt{N}\right)$, $\forall i\in \{j\in \mathcal{V}: x_j\notin \Omega_j^{\epsilon}\}$, which implies $d_{\Omega_i^{\epsilon}}(x_i)\leq \nu/\sqrt{N}$, $\forall i\in \mathcal{V}$, implying $d_{A^{\epsilon}}(\mathbf{x})\leq \nu$. Therefore, we set $c_1 = e^{-\frac{1}{\beta_1\cdot \nu/\sqrt{N}}}\cdot \left(\beta_1 \nu/\sqrt{N}\right)$.

Next, we construct a constant $c_2$. Based on the fact that $0\leq h_i(x_i)\leq 1$  and $\beta_1 d_{\Omega_i^{\epsilon}}(x_i)\leq V_i(x_i)\leq \beta_2 d_{\Omega_i^{\epsilon}}(x_i) $, we have 
    $V(\mathbf{x})\leq \sum_{i=1}^{N}V_i(x_i)\leq \sum_{i=1}^{N}\beta_2 \cdot d_{\Omega_i^{\epsilon}}(x_i)$. 
By Cauchy–Schwarz, we have $\sum_{i=1}^{N}\beta_2  d_{\Omega_i^{\epsilon}}(x_i(t)) \leq \beta_2\sqrt{N} \sqrt{d_{\Omega_1^{\epsilon}}^2(x_1) + \cdots + d_{\Omega_N^{\epsilon}}^2(x_N)}  = \beta_2\sqrt{N} d_{A^{\epsilon}}(\mathbf{x})\leq  \beta_2\sqrt{N}R$.
It follows that 
\begin{align*}
    d_{A^{\epsilon}}(\mathbf{x})\leq R \implies V(\mathbf{x})\leq \beta_2\sqrt{N} d_{A^{\epsilon}}(\mathbf{x})\leq  \beta_2\sqrt{N}R = c_2.
\end{align*}
Therefore, we set $c_2 = \beta_2\sqrt{N}R$.
If $c_1\geq c_2$, by $\dot{V}(\mathbf{x})\leq 0$ for all $\mathbf{x}\in \{\mathbf{x}\in \mathcal{X}: V(\mathbf{x})\leq c_2\}$ and $V(\mathbf{x}(t_0))\leq c_2$, we have $V(\mathbf{x}(t))\leq c_1$ always holds, and  $d_{A^{\epsilon}}(\mathbf{x})\leq \nu$ holds for all time. Otherwise, if $c_1<c_2$, we define $\Gamma = \{\mathbf{x}: c_1\leq V(\mathbf{x})\leq c_2\}$, which is a compact set. Let $-\gamma(R, \nu) = \max_{\mathbf{x}\in \Gamma} \dot{V}(\mathbf{x})$, which satisfies $-\gamma(R, \nu)<0$. Defining 
$\mathcal{T}(R, \nu) \triangleq (c_2-c_1)/\gamma (R, \nu)$, 
we have that for all $t\geq \mathcal{T}(R, \nu)$, the system state satisfies $d_{A^{\epsilon}}(\mathbf{x}(t))\leq \nu$. We prove this result by contradiction.
Suppose that there exists a $t\geq \mathcal{T}(R, \nu)$, such that $d_{A^{\epsilon}}(\mathbf{x}(t))> \nu$. Then, by construction of $c_1$, it follows that $V(\mathbf{x}(t))>c_1$. However, we have
\begin{align*}
    &V(\mathbf{x}(t)) = V(\mathbf{x}(0)) + \int_{0}^{t}\dot{V}(\mathbf{x}(t))dt\\
    & \leq c_2+ \mathcal{T}(R, \nu)\dot{V}(\mathbf{x}(t))\leq c_2+\mathcal{T}(R, \nu)(-\gamma(R, \nu)) = c_1,
\end{align*}
which leads to a contradiction. Therefore, for all $t\geq \mathcal{T}(R, \nu)$, $d_{A^{\epsilon}}(\mathbf{x}(t))\leq \nu$. This completes the proof. \hfill$\blacksquare$

\subsection{Verification of Assumptions HF, H2, and H3}
\label{appendix:lemmas}
To establish Theorem \ref{theorem:safe-stable}, we first present several auxiliary results. By invoking Theorem \ref{:theorem-converge}, convergence follows once Assumptions HF, H1, H2, and H3 are verified for system \eqref{eq:two-scale-dynamics}. Assumptions HF, H2, and H3 are established through the following three lemmas, each corresponding to one assumption.
Define $\tilde{\mathbf{x}}=\text{col}(\mathbf{x}, \mathbf{w}, \mathbf{z}, \mathbf{y})$. The dynamics in \eqref{eq:two-scale-dynamics} admit the representation $\dot{\tilde{\mathbf{x}}} = \mathcal{F}(\tilde{\mathbf{x}})$.
\begin{lemma}
\label{lemma:hf}
    If Assumptions \ref{assumption:connectivity} to \ref{assumption:locally-lipschitz} hold, then \label{lemma:hf-assumption}
    given $\mathbf{\tilde{x}}$, $\{\mathcal{F}(\tilde{\mathbf{x}})\}$ has nonempty convex compact values and is L-Lipschitz. 
\end{lemma}
\begin{proof}
     The closed-loop dynamics  \eqref{eq:two-scale-dynamics} define a dynamical system. Assumptions \ref{assumption:compute_u_bar} and \ref{assumption:feasible-mismatch} guarantee that $\overline{u}_i$ exists and is instantly accessible at time $t$ for all $\mathbf{x}\in \mathbb{R}^{nN}$ and $\mathbf{z}^{\{i\}}\in \mathbb{R}^{4N-2}$,
     and hence $\mathcal{F}(\tilde{\mathbf{x}})$ exists.
     The set $\{\mathcal{F}(\tilde{\mathbf{x}})\}$ is a singleton set, hence, it has nonempty convex compact values.       
     Since a valid CBF is continuously differentiable, by Assumption \ref{assumption:connectivity}, we have  ${g}_{i}^{i,j}(\cdot)$, ${g}_{\text{obs}}^{i,k}(\cdot)$, and ${g}_{i}^{c}(\cdot)$ are locally Lipschitz functions in $(\mathbf{x}, \mathbf{z}^{\{i\}})$. Moreover, $h_i(x_i)V_i(x_i)$ is also twice continuously differentiable in $x_i$, hence, $\Psi_{i}(\mathbf{x})$ is locally Lipschitz in $\mathbf{x}$.    
    Considering the Lipschitz continuity of the $\max$ operator and that of $\overline{u}_i$, which is guaranteed by Assumption \ref{assumption:locally-lipschitz}, it follows that the condition in HF is satisfied.
\end{proof}
Define $(\mathbf{z}^{\{i\}})^{\ast} = \text{col}(\{(z_i^{i,j})^{\ast}, (z_j^{i,j})^{\ast}:j\neq i\}, (\mathbf{z}^{c})^{\ast}, (\mathbf{z}^{clf})^{\ast})$, which consists of equilibrium components of $\mathbf{z}^{\ast}$ corresponding to the mismatch variables communicated by agent $i$. Next, we show that the designed controller fulfills Assumption H2. 
\begin{lemma}\label{lemma:h2}
    There exists a nonempty compact set $A^{\epsilon}$ which is UGAS for the system $\dot{x}_i = F_i(x_i)+G_i(x_i)\overline{u}_i(\mathbf{x}, (\mathbf{z}^{\{i\}})^{\ast})$.
\end{lemma}
\begin{proof}
Given $\mathbf{z}^{\ast}$, we have $\overline{\mathbf{u}}(\mathbf{x}, \mathbf{z}^{\ast}) = \mathbf{u}^{\ast}$, and $(\overline{\mathbf{u}}(\mathbf{x},\mathbf{z}^{\ast}), \mathbf{z}^{\ast})$ is the optimal solution to \eqref{eq:recast-problem-formulation}. By Proposition \ref{propisition-equivalence}, we have $\overline{\mathbf{u}}(\mathbf{x}, \mathbf{z}^{\ast})$ as the optimizer of \eqref{eq:problem-formulation}. Based on the results in Theorem \ref{theorem:clf-convergence}, the set $A^{\epsilon}$ is UGAS for the system $\dot{x}_i = F_i(x_i)+G_i(x_i)\overline{u}_i(\mathbf{x}_{\mathcal{N}_i}, (\mathbf{z}^{\{i\}})^{\ast})$. 
\end{proof}
Finally, we verify Assumption H3, which ensures the required convergence of the closed-loop system trajectories to the target region. Let $\text{col}(\mathbf{w}_0, \mathbf{z}_0, \mathbf{y}_0)$ denote the initial state of fast variables $\text{col}(\mathbf{w}, \mathbf{z}, \mathbf{y})$ in \eqref{eq:two-scale-dynamics}. Define $E(\mathbf{x}) = \{\text{col}(\mathbf{w}^{\ast}(\mathbf{x}), \mathbf{z}^{\ast}(\mathbf{x}), \mathbf{y}^{\ast}(\mathbf{x}))\}$ and $\tilde{\mathbf{z}}(t) = \text{col}(\mathbf{w}(t), \mathbf{z}(t), \mathbf{y}(t))$.
Define the matrices
\begin{align*}
    \Psi(\mathbf{x}) &\triangleq \text{blkdiag}\!\big(\Psi_1(\mathbf{x}), \ldots, \Psi_N(\mathbf{x})\big)\\
    \Theta(\mathbf{x}) &\triangleq \text{blkdiag}\!\big(\Theta_1(\mathbf{x}), \ldots, \Theta_N(\mathbf{x})\big)
\end{align*}
where \(\text{blkdiag}(\cdot)\) denotes the block-diagonal concatenation of its matrix arguments. Also define
\begin{align*}
    \phi(\mathbf{x}) \triangleq \text{col}(\phi_1(\mathbf{x}), \ldots, \phi_N(\mathbf{x})).
\end{align*}
Then the stacked constraint vector can be written as
\begin{align*}
    g(\mathbf{w},\mathbf{z};\mathbf{x}) &\triangleq \text{col}\!\big(g_1(w_1; \mathbf{z}^{\{1\}},\mathbf{x}), \ldots, g_N(w_N; \mathbf{z}^{\{N\}},\mathbf{x})\big)\\
    &= \Psi(\mathbf{x}) \mathbf{w} + \Theta(\mathbf{x}) \mathbf{z} + \phi(\mathbf{x}).
\end{align*}
\begin{lemma}
    \label{lemma:h3}
    The set $E(\mathbf{x})$ is UGAS for the system \eqref{eq:fast-1} to \eqref{eq:fast-last}, with functions $\psi_{R_1}$ and $T_{R_1}$ defined as in Definition \ref{def:ugas}, such that, $\forall R_1>0$, $\forall \mathbf{x}\in \overline{B}_{R_1}$, $\forall r_1>0$, $\forall \text{col}(\mathbf{w}_0, \mathbf{z}_0, \mathbf{y}_0)\in E(\mathbf{x}) + \overline{B}_{r_1}$, $\forall \nu_1>0$, we have
    \begin{enumerate}
        \item $\lim_{t\rightarrow \infty}d_{E(\mathbf{x})}(\tilde{\mathbf{z}}(t))=0$
        \item $d_{E(\mathbf{x})}(\tilde{\mathbf{z}}(t))\leq \psi_{R_1}(r_1), \forall t\geq t_0$
        \item $d_{E(\mathbf{x})}(\tilde{\mathbf{z}}(t))\leq \nu, \forall t\geq T_{R_1}(R_1, \nu_1)$
    \end{enumerate}
\end{lemma}
\begin{proof}
Since the cost function is strictly convex and constraints are convex in Problem \eqref{eq:recast-problem-formulation}, we have that $E(\mathbf{x})$ is a singleton set. Statement (1) follows directly from Theorem~\ref{theorem:5-1}, which establishes that the equilibrium set $E(\mathbf{x})$ is asymptotically stable for the fast dynamics for every fixed value of the slow variable $\mathbf{x}$. Therefore,
$\lim_{t\to \infty} d_{E(\mathbf{x})}(\tilde{\mathbf{z}}(t)) = 0$.

Next, we verify statement (2) by establishing the uniform bound $d_{E(\mathbf{x})}(\tilde{\mathbf{z}}(t))\leq m_{R_1}(r_1), \forall t\geq t_0$. 
    Let $\hat{\mathbf{w}}\triangleq\text{col}(\mathbf{w}, \mathbf{z})$ with initial condition $\hat{\mathbf{w}}_0 = \text{col}(\mathbf{w}_0, \mathbf{z}_0)$, and $M = \sum_{i=1}^{N}M_i$. Define 
    \begin{align*}
        \mathcal{L}(\hat{\mathbf{w}}, \mathbf{y}) \triangleq \frac{(\|\mathbf{w}\|^2+\xi \|\mathbf{z}\|^2)}{2} + \mathbf{y}^T(\Psi(\mathbf{x})\mathbf{w}+\Theta(\mathbf{x})\mathbf{z}+\phi(\mathbf{x})).
    \end{align*}
    We set 
    \begin{align*}
        W(\hat{\mathbf{w}}, \mathbf{y})&\triangleq \frac{1}{2}(d_{E(\mathbf{x})}(\tilde{\mathbf{z}}(t)))^{2} + \frac{1}{2}\|\nabla_{\hat{\mathbf{w}}}\mathcal{L}(\hat{\mathbf{w}}, \mathbf{y})\|^2\\
        &+\frac{1}{2} \sum_{j \in\{1, \ldots, M\}\setminus \mathcal{J}}{((\nabla_{\mathbf{y}}\mathcal{L}(\hat{\mathbf{w}},\mathbf{y}))^{(j)})^{2}}
    \end{align*}
    where 
\begin{align*}
    \mathcal{J} &= \{j\in \{1, \ldots, M\}: y_j = 0 \text{\ and\ } (\nabla_{\mathbf{y} }F(\hat{\mathbf{w}}, \mathbf{y}))^{(j)}<0\}.
\end{align*}
In what follows, to simplify notation, we omit the explicit dependence on $\mathbf{x}$ and write ${\mathbf{w}}^{\ast}$, ${\mathbf{z}}^{\ast}$,  $\hat{\mathbf{w}}^{\ast}$ and ${\mathbf{y}}^{\ast}$ instead of ${\mathbf{w}}^{\ast}(\mathbf{x})$, ${\mathbf{z}}^{\ast}(\mathbf{x})$, $\hat{\mathbf{w}}^{\ast}(\mathbf{x})$, and ${\mathbf{y}}^{\ast}(\mathbf{x})$, respectively.
We have that $d_{E(\mathbf{x})} (\begin{bmatrix}
    \hat{\mathbf{w}}(t)\\ \mathbf{y}(t)
\end{bmatrix})^2 =\|\begin{bmatrix}
     \hat{\mathbf{w}}(t) -  \hat{\mathbf{w}}^{\ast}\\
    \mathbf{y}(t) - \mathbf{y}^{\ast}
\end{bmatrix}\|^2$, which implies
    $d_{E(\mathbf{x})} (\begin{bmatrix}
    \hat{\mathbf{w}}(t)\\ \mathbf{y}(t)
\end{bmatrix})^2 \leq 2W(\hat{\mathbf{w}}(t), \mathbf{y}(t))$.
Next, we derive an upper bound of $W(\hat{\mathbf{w}}, \mathbf{y})$.
Since $\text{col}(\hat{\mathbf{w}}^{\ast},\mathbf{y}^{\ast})$ is the saddle point of $ \mathcal{L}(\hat{\mathbf{w}}, \mathbf{y})$, we have that 
\begin{align*}
  &\nabla_{\hat{\mathbf{w}}}
    \mathcal{L}(\hat{\mathbf{w}}, \mathbf{y}) = \begin{bmatrix}
        \nabla_{{\mathbf{w}}}
    \mathcal{L}(\hat{\mathbf{w}}, \mathbf{y})^T\\
    \nabla_{{\mathbf{z}}}
    \mathcal{L}(\hat{\mathbf{w}}, \mathbf{y})^T
    \end{bmatrix}^T = \begin{bmatrix}
        (\mathbf{w} + {\Psi}(\mathbf{x})^T\mathbf{y})^T\\
        ( \xi\mathbf{z} + {\Theta}(\mathbf{x})^T\mathbf{y})^T
    \end{bmatrix}^T \\
    &=\begin{bmatrix}
        (\mathbf{w}- \mathbf{w}^{\ast}+ {\Psi}(\mathbf{x})^T(\mathbf{y}-\mathbf{y}^{\ast}))^T\\
         (\xi(\mathbf{z}-\mathbf{z}^{\ast}) + {\Theta}(\mathbf{x})^T(\mathbf{y}-\mathbf{y}^{\ast}))^T
    \end{bmatrix}^T.
\end{align*}
It follows that
\begin{align*}
  &\|\nabla_{\hat{\mathbf{w}}}
    \mathcal{L}(\hat{\mathbf{w}}, \mathbf{y})\|^2  \leq 2||\mathbf{w}-\mathbf{w}^{\ast}||^{2} + 2\xi^2||\mathbf{z}-\mathbf{z}^{\ast}||^2 \\
    &+ 2\left(\sup_{\mathbf{x}\in \overline{B}_{R_1}}{||\Psi(\mathbf{x})||^{2}}\right)||\mathbf{y}-\mathbf{y}^{\ast}||^{2} \\
    &+   2\left(\sup_{\mathbf{x}\in \overline{B}_{R_1}}{||\Theta(\mathbf{x})||^{2}}\right)||\mathbf{y}-\mathbf{y}^{\ast}||^{2}.
\end{align*}
Similarly, we have 
\begin{align*}
   &\sum_{j \notin \mathcal{J}}{((\nabla_{\mathbf{y}}\mathcal{L}(\hat{\mathbf{w}},\mathbf{y}))_{j})^{2}} \leq ||\nabla_{\mathbf{y}}\mathcal{L}(\hat{\mathbf{w}},\mathbf{y})||^{2} \\
                &= ||(\mathbf{w}-\mathbf{w}^{\ast})^{T}\Psi(\mathbf{x})+ (\mathbf{z}-\mathbf{z}^{\ast})^T\Theta(\mathbf{x})||^{2} \\
                & \leq 2\sup_{\mathbf{x}\in \overline{B}_{R_1}}{||\Psi(\mathbf{x})||^{2}}||\mathbf{w}-\mathbf{w}^{\ast}||^{2} \\
                &\quad\quad\quad\quad\quad\quad\quad+ 2\sup_{\mathbf{x}\in \overline{B}_{R_1}}{||\Theta(\mathbf{x})||^{2}}||\mathbf{z}-\mathbf{z}^{\ast}||^{2}. 
\end{align*}
By definition, $d_{E(\mathbf{x})} (\begin{bmatrix}
    \hat{\mathbf{w}}(t)\\ \mathbf{y}(t)
\end{bmatrix})^2 = ||{\mathbf{w}}-{\mathbf{w}}^{\ast}||^{2} + ||{\mathbf{z}}-{\mathbf{z}}^{\ast}||^{2}+||\mathbf{y}-\mathbf{y}^{\ast}||^{2}$.             Combining terms, we have 
            \begin{align*}
               &2W(\hat{\mathbf{w}}, \mathbf{y}) \leq \left(3 + 2\sup_{\mathbf{x}\in \overline{B}_{R_1}}{||\Psi(\mathbf{x})||^{2}}\right)||\mathbf{w}-\mathbf{w}^{\ast}||^{2} \\
               & + \left(1+2\xi^2 + 2\sup_{\mathbf{x}\in \overline{B}_{R_1}}||\Theta(\mathbf{x})||^{2}\right)||\mathbf{z}-\mathbf{z}^{\ast}||^2\\
               &+ \left(1+2\sup_{\mathbf{x}\in  \overline{B}_{R_1}}{||\Psi(\mathbf{x})||_{2}^{2}}+2\sup_{\mathbf{x}\in  \overline{B}_{R_1}}{||\Theta(\mathbf{x})||^{2}}\right)||\mathbf{y}-\mathbf{y}^{\ast}||^{2},          
            \end{align*} and we can choose 
            \begin{align*}
                k_1(R_1)& = \max\big\{3 + 2\sup_{\mathbf{x}\in \overline{B}_{R_1}}{||\Psi(\mathbf{x})||^{2}}, \\
                &1 + 2\xi^2 + 2\sup_{\mathbf{x}\in \overline{B}_{R_1}}||\Theta(\mathbf{x})||^{2}, \\
                &1+2\sup_{\mathbf{x}\in \overline{B}_{R_1}}{||\Psi(\mathbf{x})||^{2}}+2\sup_{\mathbf{x}\in  \overline{B}_{R_1}}{||\Theta(\mathbf{x})||^{2}}\big\}.
            \end{align*}
with $k_1(R_1)>0$. Besides, based on Theorem \ref{theorem:5-1}, we have $W(\hat{\mathbf{w}}(t_1), \mathbf{y}(t_1)) \geq W(\hat{\mathbf{w}}(t_2), \mathbf{y}(t_2))$ for all $t_0\leq t_1\leq t_2$.
 This implies
 \begin{align*}
     &d_{E(\mathbf{x})} (\begin{bmatrix}
    \hat{\mathbf{w}}(t)\\ \mathbf{y}(t)
\end{bmatrix})^2\leq 2 W(\hat{\mathbf{w}}, \mathbf{y})\leq 2W(\hat{\mathbf{w}}(t_0), \mathbf{y}(t_0))\\
&\quad\quad \leq k_1(R_1)\|\begin{bmatrix}
    \hat{\mathbf{w}}(t_0)-\hat{\mathbf{w}}^{\ast}(\mathbf{x})\\ \mathbf{y}(t_0)-\mathbf{y}^{\ast}(\mathbf{x})
\end{bmatrix}\|^2 =k_1(R_1)r_1^2.
 \end{align*}
 For fixed $R_1$, we set $\psi_{R_1}(r_1) = \sqrt{k_1(R_1)} r_1$, which is nondecreasing in $r_1$ and $ \lim_{r\rightarrow0}\psi_{R_1}(r_1)=0$.     
    Hence, the uniform bound is established, completing the proof of statement (2) in Lemma~\ref{lemma:h3}. 

Finally, we verify the uniform convergence-time condition by showing that, $\forall \nu_1>0$, there exists $T_{R_1}(R_1, \nu_1)$ such that $d_{E(\mathbf{x})}(\tilde{\mathbf{z}}(t))\leq \nu_1, \forall t\geq T_{R_1}(R_1, \nu_1)$.
 Since $W(\cdot)$ is nonincreasing in time, it suffices to show that there exists $T$ such that $W(\hat{\mathbf{w}}(t),\mathbf{y}(t)) \leq \frac{1}{2} \nu_1^{2}$ for all $t\geq T$, since this would imply 
that $d_{E(\mathbf{x})} (\begin{bmatrix}
    \hat{\mathbf{w}}(t)\\ \mathbf{y}(t)
\end{bmatrix}) \leq \sqrt{2W(\hat{\mathbf{w}}(t), \mathbf{y}(t))}\leq \nu_1$.
We have constructed $\psi_{R_1}(r_1) = \sqrt{k_1(R_1)}r_1$ with $d_{E(\mathbf{x})} (\begin{bmatrix}
    \hat{\mathbf{w}}(t)\\ \mathbf{y}(t)
\end{bmatrix}) \leq \sqrt{k_1(R_1)}r_1$.
            If $\nu_1^2\geq k_1(R_1)r_1^2$, we have $\forall T\geq t_0$, $d_{E(\mathbf{x})} (\begin{bmatrix}
    \hat{\mathbf{w}}(t)\\ \mathbf{y}(t)
\end{bmatrix})\leq \nu_1$ for all $t\geq T$.          
            Otherwise, if $\nu_1^2< k_1(R_1)r_1^2$, we define the set $\Gamma_w(R_1,\nu_1,r_1)$ by
            \begin{align*}
              &\Gamma_w(R_1,\nu_1,r_1)\\
              &\triangleq \bigg\{
            \begin{bmatrix}
    \hat{\mathbf{w}}\\ \mathbf{y}
\end{bmatrix}
            : \left\|\begin{bmatrix}
    \hat{\mathbf{w}}-\hat{\mathbf{w}}^{\ast}\\ \mathbf{y}-\mathbf{y}^{\ast}
\end{bmatrix}\right\|^{2} \in \left[{\nu_1^{2}}, k_{1}(R_1)r_1^{2}\right]\bigg\}.
            \end{align*}  
            We have that $\Gamma_w(R_1,\nu_1,r_1)$ is a compact set.
Define 
\begin{align*}
    \overline{W}(\hat{\mathbf{u}}, \mathbf{y})&=\frac{1}{2}(d_{E(\mathbf{x})}\tilde{\mathbf{z}}(t))^{2} + \frac{1}{2}(\|\nabla_{\hat{\mathbf{w}}}\mathcal{L}(\hat{\mathbf{w}}, \mathbf{y})\|^2\\
        &+ \frac{1}{2}\sum_{j \in\{1, \ldots, M\}\setminus \overline{\mathcal{J}}}{((\nabla_{\mathbf{y}}\mathcal{L}(\hat{\mathbf{w}},\mathbf{y}))^{(j)})^{2}}) 
\end{align*}
where
\begin{align*}
    \overline{\mathcal{J}} &= \{j\in \{1, \ldots, M\}: y_j = 0 \text{\ and\ } (\nabla_{\mathbf{y} }F(\hat{\mathbf{u}}, \mathbf{y}))^{(j)}\leq 0\}.
\end{align*}
We have $\overline{W}(\hat{\mathbf{w}}(t), \mathbf{y}(t)) = {W}(\hat{\mathbf{w}}(t), \mathbf{y}(t))$ for all $t\geq t_0$.
Based on statement (1) in Theorem \ref{theorem:5-1}, we have that ${W}(\hat{\mathbf{w}}(t), \mathbf{y}(t))$ is differentiable almost everywhere. This implies that when $\dot{W}(\hat{\mathbf{w}}(t), \mathbf{y}(t))$ exists, $\dot{\overline{W}}(\hat{\mathbf{w}}(t), \mathbf{y}(t))$ also exists, with $\dot{W}(\hat{\mathbf{w}}(t), \mathbf{y}(t)) = \dot{\overline{W}}(\hat{\mathbf{w}}(t), \mathbf{y}(t))$.
We have that $\nabla_{\mathbf{y}}\mathcal{L}(\hat{\mathbf{w}},\mathbf{y}) = \Phi(\mathbf{x})\mathbf{w} + \Theta(\mathbf{x}) \mathbf{z}+\phi(\mathbf{x})$, where $\Phi(\mathbf{x})$, $\Theta(\mathbf{x})$, and $\phi(\mathbf{x})$ are bounded. This implies $\nabla_{\mathbf{y}}\mathcal{L}(\hat{\mathbf{w}},\mathbf{y})$ exists for all $\mathbf{x}$ and is linear in $(\hat{\mathbf{w}}, \mathbf{y})$ for a fixed $\mathbf{x}$.

For a given $\mathbf{x}$, let $\mathcal{H}_i(\hat{\mathbf{w}}; \mathbf{x}) \triangleq (\nabla_{\mathbf{y}}\mathcal{L}(\hat{\mathbf{w}},\mathbf{y}))_i$.
We define $\overline{\mathcal{S}}_i(\hat{\mathbf{w}}; \mathbf{x}) \triangleq \mathcal{H}_i^{-1}((-\infty, 0])$ as the preimage of a closed set. Hence, it is closed.
For each $a\in 2^{\{1, \ldots, M\}}$, we define $\mathcal{S}_a(\hat{\mathbf{w}}, \mathbf{y}; \mathbf{x}) \triangleq \{(\hat{\mathbf{w}}, \mathbf{y}; \mathbf{x}): y_j = 0 \text{\ and\ } (\nabla_{\mathbf{y}}\mathcal{L}(\hat{\mathbf{w}},\mathbf{y}))_j\leq0, \forall j\in a\}$. This implies that $\mathcal{S}_a(\hat{\mathbf{w}}, \mathbf{y}; \mathbf{x}) = \cap_{i\in a}(\overline{\mathcal{S}}_i(\hat{\mathbf{w}}; \mathbf{x})\times \{\mathbf{y}: y_i=0\})$ is a closed set. 
Based on the fact that $\Gamma_w(R_1,\nu_1,r_1)$ is a compact set, we have that $\Gamma_a(\hat{\mathbf{w}}, {\mathbf{y}}; \mathbf{x}) \triangleq \mathcal{S}_a(\hat{\mathbf{w}}, {\mathbf{y}}; \mathbf{x})\cap \Gamma_w(R_1, \nu_1, r_1)$ is a compact set.
Therefore, for each nonempty set $\Gamma_{a}$, such that $\overline{W}$ is differentiable on $\Gamma_{a}$, there exists $\gamma_a(R, \nu, r)$, such that 
\begin{align*}
    &\dot{W}(\hat{\mathbf{w}}(t), \mathbf{y}(t))= \dot{\overline{W}}(\hat{\mathbf{u}}(t), \mathbf{y}(t))\\
    &< -\gamma_a (R_1, \nu_1, r_1), ~\forall\text{col}(\hat{\mathbf{w}}, {\mathbf{y}}; \mathbf{x})\in \Gamma_a(\hat{\mathbf{w}}, {\mathbf{y}}; \mathbf{x}).
\end{align*}
Define 
\begin{align*}
    -\overline{\gamma}(R_1, \nu_1, r_1) = \max_{a\in 2^{\{1, \ldots, M\}}\setminus \{b\in 2^{\{1, \ldots, p\}}: \Gamma_b = \emptyset\}    } -\gamma_a(R, \nu, r).
\end{align*}
For any $t\geq t_0$ such that $W(\cdot)$ is not differentiable at $t$, we have  $\lim_{t\uparrow t^{\prime}}{W}(\hat{\mathbf{w}}(t), \mathbf{y}(t))\geq {W}(\hat{\mathbf{w}}(t^{\prime}), \mathbf{y}(t^{\prime}))$.
Therefore, we have that 
\begin{align*}
    {W}(\hat{\mathbf{w}}(t), \mathbf{y}(t)) 
    &\leq {W}(\hat{\mathbf{w}}(t_0), \mathbf{y}(t_0)) - \int_{t_0}^T\overline{\gamma} (R_1, \nu_1, r_1) dt \\
    &={W}(\hat{\mathbf{w}}(t_0), \mathbf{y}(t_0)) -\overline{\gamma} (R, \nu, r)\cdot (T-t_0).\\
\end{align*}
In order to ensure that there exists $T$ such that ${W}(\hat{\mathbf{w}}(t), \mathbf{y}(t)) \leq \frac{1}{2}\nu_1^{2}$, a sufficient condition is 
\begin{align*}
    &{W}(\hat{\mathbf{w}}(t_0), \mathbf{y}(t_0)) -\overline{\gamma} (R_1, \nu_1, r_1)\cdot (T-t_0) \\
    &\leq k_2(R_1)r_1^2 - \overline{\gamma} (R_1, \nu_1, r_1)\cdot (T-t_0) \leq \frac{1}{2}\nu_1^{2}.
\end{align*}
This implies that for all $t\geq T_{R_1}$ with $T_{R_1}(r_1, \nu_1) = \frac{k_2(R_1)r_1^2 - \frac{1}{2}\nu_1^2}{\overline{\gamma}(R_1, \nu_1, r_1)}+t_0$, 
             there exists $\overline{\gamma}(R_1,\nu_1,r_1) > 0$ such that $\dot{W}(\hat{\mathbf{w}}(t), \mathbf{y}(t)) < -\overline{\gamma}(R_1,\nu_1,r_1)$ for all $[\hat{\mathbf{w}}^T,\mathbf{y}^T]^T \in \Gamma_w(R_1,\nu_1,r_1)$ and $\mathbf{x} \in \overline{B}_{R_1}$.
             We prove this result by contradiction. Suppose $t > T_{R_1}(r_1,\nu_1)$ and yet $W(\hat{\mathbf{u}}(t),\mathbf{y}(t)) > k_{1}(R_1)\nu_1^{2}$. We have for all $t^{\prime} \in [0,t]$, when $\dot{W}(\hat{\mathbf{w}}(t^{\prime}),\mathbf{y}(t^{\prime}))$ exists, $\dot{W}(\hat{\mathbf{w}}(t^{\prime}),\mathbf{y}(t^{\prime})) < -\overline{\gamma}(R_1,\nu_1,r_1)$ , and otherwise $\lim_{t\uparrow t^{\prime}}{W}(\hat{\mathbf{w}}(t), \mathbf{y}(t))\geq {W}(\hat{\mathbf{w}}(t^{\prime}), \mathbf{y}(t^{\prime}))$, so 
            \begin{align*}
            &W(\hat{\mathbf{w}}(t), \mathbf{y}(t)) \leq k_{2}(R_1)r_1^{2} - \overline{\gamma}(R_1,\nu_1,r_1)(t-t_0) \\
            &< k_{2}(R_1)r_1^{2} - \overline{\gamma}(R_1,\nu_1,r_1)(T_{R_1}(r_1,\nu_1)-t_0) \\
            &= k_{2}(R_1)r_1^{2}-(k_{2}(R_1)r_1^{2} - \frac{1}{2}\nu_1^{2}) = \frac{1}{2}\nu_1^{2},
            \end{align*}
            a contradiction. Hence, statement (3) follows. The set $E(\mathbf{x})$ is UGAS for the system \eqref{eq:fast-1} to \eqref{eq:fast-last}. This completes the proof.
\end{proof}

\end{document}